\documentclass[12pt]{article}

\usepackage{amsmath,amssymb,subfigure,amsthm,mathtools,enumerate}
\usepackage[mathscr]{eucal}
\usepackage[pdftex]{graphicx} 
\usepackage{fullpage}

\newcommand{\R}{\mathbb{R}}
\newcommand{\C}{\mathbb{C}}
\newcommand{\Z}{\mathbb{Z}}
\newcommand{\N}{\mathbb{N}}
\newcommand{\sA}{\mathscr{A}}
\newcommand{\sM}{\mathscr{M}}
\newcommand{\sY}{\mathscr{Y}}
\newcommand{\sD}{\mathscr{D}}
\newcommand{\sG}{\mathscr{G}}

\newcommand{\psum}{\mathop{\sum\nolimits'}}
\newcommand{\pprod}{\mathop{\prod\nolimits'}}

\newcommand{\ignore}[1]{} 

\newtheorem{theorem}{Theorem}[section]
\newtheorem{lemma}[theorem]{Lemma}
\newtheorem{proposition}[theorem]{Proposition}
\newtheorem{definition}[theorem]{Definition}

\newtheorem{corollary}[theorem]{Corollary}

\title{An Optimal Family of Exponentially Accurate One-Bit Sigma-Delta
Quantization Schemes}
\author{Percy Deift\footnotemark[1] \and 
C.~Sinan G{\"u}nt{\"u}rk\footnotemark[1] \and  
Felix Krahmer\footnotemark[1] \footnotemark[2] }
\date{January 22, 2010}
\begin{document}
\renewcommand{\thefootnote}{\fnsymbol{footnote}}
\footnotetext[1]{
    Courant Institute of Mathematical Sciences, New York University, New York, NY, USA.
    }
 \footnotetext[2]{
   Hausdorff Center for Mathematics, Universit{\"a}t Bonn, Bonn, Germany.
   }
   
\renewcommand{\thefootnote}{\arabic{footnote}}
\maketitle

\begin{abstract} 
Sigma-Delta modulation is a popular method for analog-to-digital conversion of bandlimited signals
that employs coarse quantization coupled with oversampling. The standard
mathematical model for the error analysis of the method measures the performance of a given 
scheme by the rate at which the associated reconstruction error decays 
as a function of the oversampling ratio $\lambda$. It was recently 
shown that exponential accuracy of the form $O(2^{-r\lambda})$ can be achieved by appropriate one-bit 
Sigma-Delta modulation schemes. By general information-entropy arguments $r$ must be less than $1$. The current best known value for $r$ is approximately $0.088$. 
The schemes that were designed to achieve this 
accuracy employ the ``greedy'' quantization rule coupled with feedback filters
that fall into a class we call ``minimally supported''. 
In this paper, we study the minimization problem that corresponds to optimizing the error decay rate for this class of feedback filters. We solve a relaxed version of this problem exactly and provide explicit asymptotics of the solutions. From these relaxed solutions, we find asymptotically 
optimal solutions of the original problem, 
which improve the best known exponential error decay 
rate to $r \approx 0.102$. 
Our method draws from the theory of orthogonal polynomials; in particular,
it relates the optimal filters to the zero sets of 
Chebyshev polynomials of the second kind.

\end{abstract}

\section{Introduction} 

Conventional Analog-to-Digital (A/D) conversion systems consist of two basic steps:
{\em sampling} and {\em quantization}. 
Sampling is the process of replacing the input function $(x(t))_{t \in \R}$ 
by a sequence of its sample values $(x(t_n))_{n \in \Z}$.  The sampling
instances $(t_n)$ are typically uniform, i.e., $t_n = n \tau$ for some 
$\tau > 0$. It is well known that this process incurs no loss of information if
the signal is bandlimited and $\tau$ is sufficiently small. 
More precisely, if $\mbox{supp}~\widehat x \subset [-\Omega,\Omega]$,
then it suffices to pick $\tau \leq \tau_\mathrm{crit} := \frac{1}{2\Omega}$, and the 
{\em sampling theorem} 
provides a recipe for perfect reconstruction via the formula
\begin{equation} \label{shannon}
x(t) = \tau \sum_{n \in \Z} x(n\tau) \varphi(t - n\tau),
\end{equation}
where $\varphi$ is any sufficiently localized
function such that 
\begin{equation}\label{admissible}
\widehat \varphi(\xi) = \left \{
\begin{array}{ll}
1, & |\xi| \leq \Omega, \\
0, & |\xi| \geq \frac{1}{2\tau},
\end{array}
\right.
\end{equation}
which we shall refer
to as the {\em admissibility} condition for $\varphi$.
Here, $\widehat x$ and $\widehat \varphi$ denote the Fourier transform of 
$x$ and $\varphi$, respectively. In this paper, we shall work with the 
following normalization of the Fourier transform on $\R$:
$$ \widehat f(\xi) := \int_{-\infty}^\infty f(t) e^{-2\pi i \xi t} \,\mathrm{d}t.$$
The value $\rho:=1/\tau$ is called the sampling rate, and
$\rho_\mathrm{crit} := 1/\tau_\mathrm{crit} = 2\Omega$ is called the critical
(or Nyquist) sampling rate. The {\em oversampling ratio} is defined as 
\begin{equation}\label{lambda}
\lambda:= \frac{\rho}{\rho_\mathrm{crit}}.
\end{equation}
We shall assume in the rest of the paper that the value of $\Omega$ is 
arbitrary but fixed.

The next step, quantization, is the process of discretization of the amplitude, 
which involves replacing each $x(t_n)$ by another
value $q_n$ suitably chosen from a fixed, typically finite 
set $\sA$ (the alphabet). The resulting
sequence $(q_n)$ forms the raw digital representation of the
analog signal $x$. 

The standard approach to recover an approximation $\tilde x$ to 
the original analog signal 
$x$ is to use the reconstruction formula 
(\ref{shannon}) with $x(n\tau)$ replaced by $q_n$, which
produces an error signal $e := x - \tilde x$ given by
\begin{equation} \label{err_sig}
e(t) = \tau \sum_{n \in \Z} \big (x(n\tau) - q_n \big) \varphi(t - n\tau).
\end{equation}
The quality of this approximation can then be measured by a variety of 
functional norms on the error signal. In this paper we will use the norm $\|e\|_{L^\infty}$ which is
standard and arguably the most meaningful.

Traditionally, A/D converters have been classified into two main families: {\em 
Nyquist-rate} converters ($\lambda \approx 1$)
and {\em oversampling} converters ($\lambda \gg 1$).
When $\tau = \tau_\mathrm{crit}$, the set of functions 
$\{\varphi(\cdot-n\tau) : n \in \Z\}$ forms an orthogonal system.
Therefore, a Nyquist-rate converter necessarily has to keep
$|x(n\tau)-q_n|$ small for each $n$ in order
to achieve small overall reconstruction error; this
requires that the alphabet
$\sA$ forms a fine net in the range of the signal. 
On the other hand, 
oversampling converters are not bound by this requirement because 
as $\tau$ decreases (i.e., $\lambda$ increases), the kernel of the operator
\begin{equation}\label{def_T}
T^\varphi_\tau:(c_n) \mapsto  \sum_{n \in \Z} c_n \varphi(\cdot - n\tau)
\end{equation}
gets ``bigger'' in a certain sense, and small error can be
achieved even with very coarse
alphabets $\sA$. The extreme case is a {\em one-bit quantization} scheme, which uses $\sA=\{-1,+1\}$. For the circuit engineer, 
coarse alphabets mean low-cost analog hardware because increasing the sampling 
rate is cheaper than refining the quantization.
For this reason, oversampling data converters, in particular, Sigma-Delta ($\Sigma\Delta$) modulators (see section~\ref{noisesh} below) have become more popular than
Nyquist-rate converters for low to medium-bandwidth signal classes, such as 
audio \cite{NST96}. 
On the other hand, the mathematics of oversampled A/D conversion is highly challenging as the selection problem of $(q_n)$ shifts from 
being local (i.e., the $q_n$'s are chosen independently for each $n$) to a more 
global one; a quantized assignment $q_n \in \sA$ should be
computed based not only on the current sample
value $x(t_n)$, but also taking into account sample values $x(t_k)$ and 
assignments $q_k$ in neighboring positions.
Many ``online'' quantization
applications, such as A/D conversion of audio signals, require causality,
i.e., only quantities that depend on
prior instances of time can be utilized.
Other applications, such as
digital halftoning, may not be strictly bound by the same kind of causality 
restrictions although it is still useful to process samples in
some preset order. In both situations, the amount of memory that can be
employed in the quantization algorithm is one of the
limiting factors determining the performance of the algorithm. 

This paper is concerned with the approximation theory of oversampled,
coarse quantization, in particular, one-bit quantization of bandlimited
functions. 
Despite the vast engineering literature on the subject (e.g., see \cite{NST96}), 
and a recent series
of more mathematically oriented papers (e.g., 
\cite{DaubDV03, yilmaz02, Gunt03,Gunt04,GunturkNguyen,BPY, BoPa07}),
the fundamental question of how to carry out optimal quantization remains open. 
After the pioneering work of Daubechies and DeVore on the
mathematical analysis and design of $\Sigma\Delta$ modulators \cite{DaubDV03}, 
more recent work showed that exponential
accuracy in the oversampling ratio $\lambda$ 
can be achieved by appropriate one-bit $\Sigma\Delta$ modulation schemes
\cite{Gunt03}.
The best achievable error decay rate for these schemes was $O(2^{-r\lambda})$ 
with $r \approx 0.076$ in \cite{Gunt03}. Later, with a modification,
this rate was improved
to $r \approx 0.088$ in \cite{Kra07}. It is known that any one-bit quantization 
scheme has to obey $r < 1$, whereas it is not known if this upper bound
is tight \cite{CalDaub02,Gunt03}. This paper improves the best achievable rate 
further to 
$r \approx 0.102$ by designing an optimal family of Sigma-Delta modulation 
schemes within
the class of so-called ``minimally supported'' recursion filters. 
These schemes were introduced in \cite{Gunt03} together with a number of accompanying open
problems. The {\bf main results} of this paper (see Theorems {\bf \ref{mainthm}, \ref{thmasop}, \ref{decayrate}})
are based on the solution of one of these problems, namely, the optimization of the minimally
supported recursion filters that are used in conjunction with the
{\em greedy} rule of quantization. 
One of our main results is that the supports of these optimal 
filters are given by suitably scaled 
zero sets of Chebyshev polynomials of the second kind, and we use this result
to derive our improved exponent for the error bound.

The paper is organized as follows. In Section \ref{background}, we review
the basic mathematical theory of $\Sigma\Delta$ modulation as well as
the constructions and methods of \cite{Gunt03} which will be relevant to this paper, such 
as the family of minimally supported recursion filters, the greedy
quantization rule, and fundamental aspects of the error analysis leading to exponentially decaying error bounds. As we explain, the discrete optimization problem introduced in \cite{Gunt03} plays a key role in the analysis. In Section \ref{prob_setup}, we introduce a relaxed version of the optimization problem, which is analytically tractable. This relaxed
problem is solved in Section \ref{solution_relax} and analyzed asymptotically in Section \ref{asymp} as the order of the reconstruction filter goes to infinity. Section \ref{asymp} also contains the construction of asymptotically optimal solutions to the discrete optimization problem, which yields the exponential error decay rate mentioned above. Finally, we extend our results to multi-level
quantization alphabets in Section \ref{multi-level-section}.
We also collect, separately in the Appendix,
the properties and identities for Chebyshev polynomials which are used in the proofs of our results.


\section{Background on $\Sigma\Delta$ modulation}
\label{background}

\subsection{Noise shaping and feedback quantization\label{noisesh}}

$\Sigma\Delta$ modulation is a generic name for a family of recursive
quantization algorithms that utilize the concept of ``noise shaping''. (The origin of the terminology $\Sigma\Delta$, or alternatively $\Delta\Sigma$, goes back to the patent application of Inose et al \cite{IYM62} and refers to the presence of certain circuit components at the level of A/D circuit implementation;
see also \cite{NST96, SchTe04}.)
Let $(y_n)$ be a general sequence to be quantized (for example, $y_n = x(n\tau)$), 
$(q_n)$ denote the quantized representation of this sequence, and
$(\nu_n)$ be the quantization error sequence (i.e., the ``noise''), defined by 
$\nu := y - q$.
Noise shaping is a quantization strategy whose objective is
to arrange for the quantization noise $\nu$ to
fall outside the frequency band of interest, which, in our case is  
the low-frequency band. Note that 
the effective error we are interested in is $e = T^\varphi_\tau(\nu)$,
hence one would like $\nu$ to be close to the kernel of $T^\varphi_\tau$.
It is useful to think of $T^\varphi_\tau(\nu)$ as a generalized convolution
of the sequence $\nu$ and the function $\varphi$ (sampled at the scale 
$\tau$). Let us introduce the notation
\begin{equation}\label{gen_conv}
\nu \circledast_\tau \varphi := T^\varphi_\tau(\nu).
\end{equation}
Note that $(a*b)\circledast_\tau \varphi
= a \circledast_\tau (b \circledast_\tau \varphi)$ and
$a \circledast_\tau (\varphi * \psi) = (a \circledast_\tau \varphi) * \psi$, 
etc., where $a$ and $b$ are sequences on $\Z$, $\varphi$ and $\psi$ are functions on
$\R$, and $*$ denotes the usual convolution operation (on $\Z$ and $\R$). 
By taking the Fourier transform of (\ref{def_T}) one sees that
the kernel of $T^\varphi_\tau$ consists of sequences $\nu$ that are 
spectrally disjoint from $\varphi$, i.e.,
``high-pass'' sequences, since $\varphi$ is a ``low-pass'' function,
as apparent from (\ref{admissible}). Thus, arranging for $\nu$
to be (close to) 
a high-pass sequence is the primary objective of $\Sigma\Delta$ 
modulation.

High-pass sequences have the property that their Fourier
transforms vanish at zero frequency (and possibly in a neighborhood). 
For a finite high-pass sequence $s$,
the Fourier transform $\widehat s(\xi) := \sum s_n 
e^{2\pi i n \xi}$ has a factor 
$(1-e^{2\pi i\xi})^m$ for some positive integer $m$,
which means that $s = \Delta^m w$ for some finite sequence
$w$. Here, $\Delta$ denotes the finite difference operator defined by
\begin{equation}\label{def_Delta}
(\Delta w)_n:=w_n - w_{n-1}.
\end{equation}
The quantization error sequence $\nu=y-q$, however, need not 
be finitely supported, and therefore 
$\widehat \nu$ need not be defined as a function (since $\nu$ is
bounded at best). Nevertheless, this spectral factorization
can be used to model more general high-pass sequences.
Indeed, a $\Sigma\Delta$ modulation scheme of order $m$
utilizes the difference equation
\begin{equation}\label{sig_del}
y-q = \Delta^m u
\end{equation}
to be satisfied for each input $y$ and its quantization $q$,
for an appropriate auxiliary sequence $u$ (called the state sequence). 
This explicit
factorization of the quantization error is useful if
$u$ is a bounded sequence, as will be explained in more detail
in the next subsection.

In practice, (\ref{sig_del}) is used as part of a quantization algorithm.
That is, given any sequence $(y_n)_{n \geq 0}$, its quantization 
$(q_n)_{n \geq 0}$ is generated
by a recursive algorithm that satisfies (\ref{sig_del}). This is 
achieved via an associated ``quantization rule'' 
\begin{equation}\label{Q_rule}
q_n = Q(u_{n-1},u_{n-2},\dots,y_n,y_{n-1},\dots),
\end{equation}
together with the ``update rule"
\begin{equation}\label{u_update}
u_n = \sum_{k=1}^m (-1)^{k-1} \binom{m}{k} u_{n-k} + y_n - q_n,
\end{equation}
which is a restatement of (\ref{sig_del}). Typically one employs
the initial conditions $u_n = 0$ for $n < 0$. 

In the electrical engineering literature, such a recursive procedure 
for quantization is
called ``feedback quantization" due to the role $q_n$ plays as a 
(nonlinear) feedback term via (\ref{Q_rule}) for the 
difference equation (\ref{sig_del}). Note that if $y$ and $q$
are given unrelated sequences and $u$ is determined from $y-q$ via
(\ref{sig_del}), 
then $u$ would typically be unbounded for $m \geq 1$.
Hence the role of the quantization rule $Q$ is to tie $q$ to $y$
in such a way as to control $u$.

A $\Sigma\Delta$ modulator may also arise from
a more general difference equation of the form 
\begin{equation}\label{sig_del_H}
y-q = H*v
\end{equation}
where $H$ is a {\em causal} sequence (i.e., $H_n = 0$ for $n < 0$)
in $\ell^1$ with $H_0 = 1$. 
If $H = \Delta^m g$ with $g \in \ell^1$, then any (bounded) solution $v$ of 
(\ref{sig_del_H}) gives rise to a (bounded) solution $u$ of (\ref{sig_del})
via $u = g*v$. Thus (\ref{sig_del_H}) can be rewritten in the {\em canonical}
form (\ref{sig_del}) by a change of variables. Nevertheless, there are 
significant advantages to working directly with representations of the form (\ref{sig_del_H}),
as explained in Section \ref{optimal} below.

\subsection{Basic error estimates\label{basicerror}}

Basic error analysis of $\Sigma\Delta$ modulation
only relies on the {\em boundedness} of solutions $u$ of 
(\ref{sig_del}) for given $y$ and $q$,
and not on the specifics of the quantization rule $Q$
that was employed. It is useful, however, to consider arbitrary
 quantizer maps $\sM : y \mapsto q$ that satisfy
(\ref{sig_del}) for some $m$ and $u$, where $u$ may or may not be
bounded. 
Note that if $u$ is a solution to (\ref{sig_del}) for a given triple
$(y,q,m)$, then $\tilde u := \Delta u$ is a solution for the triple
$(y,q,m{-}1)$. Hence a given map $\sM$ can be treated
as a $\Sigma\Delta$ modulator of different orders. 
Formally, we will refer to the pair $(\sM,m)$ as a $\Sigma\Delta$ modulator
of order $m$. 

As indicated above, in order for 
a $\Sigma\Delta$ modulator $(\sM,m)$ to be useful, there must be
a bounded solution $u$ to (\ref{sig_del}). (Note that, up to an additive
constant, there can be at most
one such bounded solution once $m > 0$.) Moreover, one would like this
to be the case for all input sequences $y$ in a given class
$\sY$, such as $\sY_\mu = \{y : ||y||_{\ell^\infty} \leq \mu\}$ for some $\mu>0$. 
In this case, we say that $(\sM,m)$ is {\em stable} 
for the input class $\sY$.
Clearly, if $(\sM,m)$ is stable for $\sY$, then $(\sM,m{-}1)$
is stable for $\sY$ as well. To any quantizer map
$\sM$ and a class of inputs $\sY$,
we assign its maximal order $m^*(\sM,\sY)$ via
\begin{equation}
m^*(\sM,\sY) := \sup \big \{ m : \forall y \in \sY,~~
\exists u \in \ell^\infty \mbox{ such that } y - \sM(y) = 
\Delta^m u \big \}.
\end{equation}
Note that both $0$ and $\infty$ are admissible values for 
$m^*(\sM,\sY)$. With this notation, 
$(\sM,m)$ is stable for the class $\sY$ if and only if 
$m \leq m^*(\sM,\sY)$. 

Stability is a crucial property. Indeed, 
it was shown in \cite{DaubDV03} that a stable $m$-th order scheme
with bounded solution $u$ results in the error bound
\begin{equation} \label{err1}
\|e \|_{L^\infty}  \leq  \|u\|_{\ell^\infty} \|\varphi^{(m)} \|_{L^1}\tau^m,
\end{equation}
where $\varphi^{(m)}$ denotes the $m$th order derivative of 
$\varphi$.
The proof of (\ref{err1}) employs repeated
summation by parts, giving rise to the commutation relation
\begin{equation}\label{commute}
(\Delta^m u) \circledast_\tau \varphi =  u \circledast_\tau 
(\Delta^m_\tau \varphi),
\end{equation}
where $\Delta_\tau$
is the finite difference operator at scale $\tau$ defined by
$(\Delta_\tau \varphi)(\cdot) := \varphi(\cdot)-\varphi(\cdot-\tau)$.
The left hand side of (\ref{commute}) is simply equal to the error signal
$e$ by definition.
On the other hand, the right hand side of this relation immediately leads  to the error bound (\ref{err1}). As $u$ is not uniquely determined by the equation
$ y-\sM(y)=\Delta^m u$, 
it is convenient to introduce the notation
\begin{equation}
U(\sM,m,y) := \inf \big \{
\|u\|_{\ell^\infty} : y - \sM(y) = \Delta^m u \big \},
\end{equation}
and for an input class $\sY$
\begin{equation}
U(\sM,m,\sY) := \sup_{y \in \sY} U(\sM,m,y).
\end{equation}
We are interested in applying (\ref{err1}) to sequences $y$ that arise as  samples $y_n = x\left(n\tau\right)$ of a bandlimited signal $x$ as above. To compare the error bounds for different values of $\tau$, one could consider the class $\sY^{(x)}=\left\{y=(y_n)_{n\in\Z}: y_n= x(n\tau) \text{ for some } \tau\right\}$ and work with the constant $U(\sM,m,\sY^{(x)})$. However, it is difficult to estimate $U(\sM,m,\sY^{(x)})$ in a way that accurately reflects the detailed nature of the signal $x$.
 Instead, we note that for $\|x\|_{L^\infty}\leq\mu$, one has $\sY^{(x)} \subset \sY_\mu$, where 
$\sY_\mu = \{y=(y_n)_{n\in\Z} : ||y||_{\ell^\infty} \leq \mu\}$ as defined above. 
This leads to the bound 
\begin{equation} \label{err1c}
\|e \|_{L^\infty}  \leq U(\sM,m,\sY_\mu) \|\varphi^{(m)} \|_{L^1}\tau^m.
\end{equation}

In (\ref{err1c}), the reconstruction kernel $\varphi$ is restricted by the $\tau$-dependent admissibility condition (\ref{admissible}). However, 
if a reconstruction kernel $\varphi_0$ is admissible 
in the sense of (\ref{admissible})
for $\tau = \tau_0
= 1/\rho_0$, then it is admissible for all $\tau < \tau_0$. 
This allows one to fix $\rho_0 = (1+\epsilon) \rho_\mathrm{crit}$
for some small $\epsilon > 0$, and set $\varphi = \varphi_0$. Now
Bernstein's inequality\footnote{If $\widehat f$ is supported in $[-A,A]$,
then $\|f'\|_{L^p} \leq 2\pi A \|f\|_{L^p}$ for $1\leq p\leq\infty$.}
 implies that $\|\varphi_0^{(m)}\|_{L^1} \leq 
(\pi \rho_0)^m \|\varphi_0\|_{L^1}$ since 
$\widehat \varphi_0$  is supported in $[-\frac{\rho_0}{2},\frac{\rho_0}{2}]$.
This estimate allows us to express the error bound naturally as a function of
the oversampling ratio $\lambda$ defined in (\ref{lambda}), as follows: 

Let $\epsilon$ and $\varphi_0$ be fixed as above, let $\sM$ be a given quantizer function, and let $m$ be a positive integer. Then for any given $\Omega$-bandlimited function $x$ with $\|x\|_{L^\infty}\leq\mu$, the quantization error $e=e_\lambda$, as expressed in (\ref{err_sig}) in terms of $\tau=\frac{1}{\lambda\rho_{crit} }$, can be bounded in terms of $\lambda$:
\begin{equation} \label{err2}
\|e_\lambda \|_{L^\infty} \leq U(\sM,m,\sY_\mu) \|\varphi_0\|_{L^1}  
\pi^m (1+\epsilon)^m \lambda^{-m}.
\end{equation}
As the constants are independent of $\lambda$, this yields an $O(\lambda^{-m})$  bound as a function of  $\lambda$. Note that 
any dependency of the error on the original bandwidth parameter $\Omega$
has now been effectively absorbed into the fixed constant 
$\|\varphi_0\|_{L^1}$. In fact, this quantity need not even 
depend on $\Omega$; a reconstruction kernel $\varphi^*_0$
can be designed once and for all corresponding to
$\Omega = 1$, and then employed to define $\varphi_0(t)
:= \Omega \varphi^*_0(\Omega t)$, which is admissible
and has the same $L^1$-norm as $\varphi^*_0$.

\subsection{Exponentially accurate $\Sigma\Delta$ modulation}
\label{optimal}

The $O(\lambda^{-m})$ error decay rate derived in the previous section is based on using a fixed $\Sigma\Delta$ modulator. It is possible, however, to improve the bounds by choosing the modulator adaptively as a function of $\lambda$. In order to obtain an error decay 
rate better than polynomial, one needs an infinite family 
$((\sM_m,m))_1^\infty$ of
stable $\Sigma\Delta$ modulation schemes from which
the optimal scheme $\sM_{m_{opt}}$ (i.e., one that yields the smallest
bound in (\ref{err2}))
is selected as a function of $\lambda$.
This point of view was first pursued systematically 
in \cite{DaubDV03}. Finding a stable 
$\Sigma\Delta$ modulator is in general a non-trivial matter as $m$
increases, and especially so if the alphabet $\sA$ is a small set.
The extreme case is one-bit $\Sigma\Delta$ modulation, i.e.,
when $\mbox{card}(\sA) = 2$. 
The first infinite family of arbitrary-order, stable one-bit $\Sigma\Delta$ 
modulators was also constructed in \cite{DaubDV03}.
In the one-bit case, one may set
$\sA = \{-1,+1\}$ to normalize the amplitude, and choose 
$\mu \leq 1$ when defining the input class $\sY = \sY_\mu$. 
The optimal order $m_{opt}$
and the size of the resulting bound on $\|e_\lambda\|_{L^\infty}$ 
 depend on the constants $U(\sM_m,m,\sY_\mu)$. 

There are a priori lower bounds on $U(\sM,m,\sY_\mu)$ for any $0<\mu\leq 1$ and any quantizer $\sM$. Indeed, as shown in 
\cite{Gunt03}, one obtains a super-exponential lower bound on $U(\sM,m,\sY_\mu)$  by considering the
average metric entropy of the space of bandlimited functions, as follows.
Define the quantity
\begin{equation}
U_m(\sY) := \inf_\sM ~U(\sM,m,\sY)
\end{equation}
and let
\begin{equation}
\mathscr{X}_\mu := \{x : \mathrm{supp}~\widehat x \in [-1/2,1/2],~
\|x\|_{L^\infty} \leq \mu \}.
\end{equation}
Then, as shown in \cite{DDGV06,Gunt03}, one has for any one-bit quantizer
\begin{equation}\label{exp_lower}
\sup \{\|e_\lambda\|_{L^\infty} : x \in \mathscr{X}_\mu \} \gtrsim_\mu 2^{-\lambda},
\end{equation}
which yields, when used with (\ref{err2}), 
the result that 
$U_m(\sY_\mu) \gtrsim_\mu (mc)^m$ for some absolute 
constant $c > 0$
\cite{Gunt03}. Here by $A\gtrsim_s B$, we mean that $A\geq C B$ for some constant $C$ that may depend on $s$, but not on any other input variables of $A$ and $B$.

For the family of $\Sigma\Delta$ modulators constructed in \cite{DaubDV03}, which we will denote by $\left(\sD_m\right)_{m=1}^\infty$, 
the best upper bounds on $U(\sD_m,m,\sY_\mu)$ are of order
$\exp(cm^2)$, resulting in the order-optimized error bound
$ \|e_\lambda\|_{L^\infty} \lesssim \exp(-c (\log \lambda)^2)$, which is substantially larger
than the exponentially small lower bound of (\ref{exp_lower}). 
The first construction that led to 
exponentially accurate $\Sigma\Delta$ modulation
was given later in \cite{Gunt03}. The associated modulators, here denoted by $\left(\sG_m\right)_{m=1}^\infty$, satisfy the 
bound
\begin{equation}\label{bound_U}
U(\sG_m,m,\sY_\mu) \lesssim (ma)^m,
\end{equation}
for $a = a(\mu) > 0$. 
Substituting (\ref{bound_U}) into (\ref{err2}) and using the elementary inequality
\begin{equation}\label{optim_ineq}
\min_{m} ~ m^m \alpha^{-m} \lesssim e^{-\alpha/e} 
\end{equation}
with $\alpha = \lambda/(a\pi(1+\epsilon))$, one obtains the order-optimized error bound
\begin{equation}\label{exp_upper}
\sup \{\|e_\lambda\|_{L^\infty} : x \in \mathscr{X}_\mu \} \lesssim 2^{-r\lambda},
\end{equation}
where $r = r(\mu) = (a \pi e  (1+\epsilon) \log 2)^{-1}$. On the other hand, the smallest achievable value of $a$ is shown to be ${6}/{e}$, which corresponds to the largest achievable value of $r=r_\mathrm{max}\left(\left(\sG_m\right)_{m=1}^\infty\right)\approx (6 \pi \log 2)^{-1} \approx 0.076$. Observe from (\ref{exp_lower}) that it is impossible to achieve an exponent $r>1$.

In \cite{Gunt03}, the $\sG_m$ were given in the form (\ref{sig_del_H}) where the causal sequences $H,g\in \ell^1$ with $H_0 = g_0 = 1$ depend on $m$, and are related via 
\begin{equation}\label{factor}
H = \Delta^m g.
\end{equation}
Define 
$h := \delta^{(0)} - H$, where 
$\delta^{(0)}$ denotes the Kronecker delta sequence supported at $0$. Then
(\ref{sig_del_H}) can be implemented recursively as
\begin{equation}\label{rec_eq}
v_n = (h*v)_n + y_n - q_n.
\end{equation}
If $q$ is chosen so that 
the resulting $v$ is bounded, i.e., if this new scheme is stable, then
$u := g*v$ is a bounded solution of (\ref{sig_del}), and
\begin{equation} \label{uvg}
\| u \|_{\ell^\infty} \leq \|g \|_{\ell^1} \|v \|_{\ell^\infty}.
\end{equation}

The significance of the more general form  (\ref{sig_del_H}) giving rise to (\ref{rec_eq}) is the following:
If 
\begin{equation} \label{hmucond}
 ||h||_1+||y||_\infty\leq 2,
\end{equation}
then the {\em greedy} quantization rule
\begin{equation}\label{greedy_Q}
 q_n := \text{sign}\left((h\ast v)_n + y_n \right).
\end{equation}
leads to a solution $v$ with $||v||_\infty \leq 1$, provided the 
initial conditions satisfy this bound. However, for the canonical form (\ref{sig_del}), the condition (\ref{hmucond}) clearly fails for all $m>1$.

For $\|y\|_\infty \leq \mu\leq 1$, a filter $h$ will satisfy (\ref{hmucond}) and hence lead to a stable scheme when $\|h\|_1 \leq 2-\mu$. Note that (\ref{factor}) together with the fact that $g\in\ell^1$ implies that $\sum h_i=1$ and hence  $\|h\|_1\geq 1$.

In view of the preceding considerations, we are led, for each $m$ to the following minimization problem:
\begin{equation}
 \text{Minimize } \|g\|_1\text{ subject to } \delta^{(0)}-h=\Delta^m g,\ \|h\|_1\leq 2-\mu. \label{minprobgen}
\end{equation}
It was shown in \cite{Gunt03} that if a sequence of filters ${h}^{(m)}$ satisfy the {\em feasibility conditions} in (\ref{minprobgen}) for the corresponding $m\in\N$, the diameter of the support set of ${h}^{(m)}$ must grow at least quadratically in $m$. 

The minimization problem (\ref{minprobgen}) was not solved in \cite{Gunt03}; rather the author introduced a class of feasible filters $h = h^{(m)}$ which were effective in the sense that they lead to an exponential error bound with rate constant $r$ as above. These 
filters $h^{(m)}$ are sparse, i.e., they contain only a few non-zero entries. Indeed, each $h^{(m)}$  has exactly $m$ non-zero entries, which can be shown to be the minimal support size for which $h^{(m)}$ can satisfy the feasibility conditions. 
We shall call such filters {\em minimally supported}. Note that if $h$ has finite support and $\|g\|_1<\infty$, then $g$ has finite support. We make the following formal definition:
\begin{definition}
We say that a filter $h=\delta^{(0)}-\Delta^m g$, for a finitely supported $g$, has {\it{minimal support}} if $|\mbox{supp } h|=m$.
\end{definition}
The goal of this paper is to
find optimal filters within the class of filters with minimal support.

\subsection{Filters with minimal support\label{minsupp}}

As the filter $\delta^{(0)}-h$ arises as the $m$-th order finite difference of the vector $g$, its entries have to satisfy $m$ moment conditions. This implies that the support size of $h$ is at least $m$, as advertised above.

For filters $h$ with minimal support
\begin{equation}
 h=\sum\limits_{j=1}^{m} d_j \delta^{({n_j})},
\end{equation}
the moment conditions lead to explicit formulae for the entries $d_j$ in terms of the support $\{n_j\}_{j=1}^m$ of $h$, where $1 \leq n_1<n_2<\dots<n_m$ \cite{Gunt03}. Here the condition that $n_1\geq 1$ follows from the strict causality of $h$.

Indeed, one finds
\begin{equation}
d_j={\pprod\limits^m_{i=1}}\frac{n_i}{n_i-n_j}.
\end{equation}

Here the notation {\footnotesize $\pprod$}, and analogously {\footnotesize $\psum$}, indicates that the singular terms are excluded from the product, or the sum respectively. By definition, if $m=1$, one has $d_1=1$.

The condition $\Vert h\Vert_1\leq 2-\mu$ then takes the form
\begin{equation}\label{L1}
 \sum\limits_{j=1}^{m} {\pprod\limits_{i=1}^{m}}\frac{n_i}{\vert n_i-n_j\vert}\leq 
2-\mu .
\end{equation}
Furthermore, explicit computations lead to the identity
\begin{equation} \label{geta}
 \Vert g\Vert_1 = \frac{\prod_{j=1}^{m}n_j}{m!}.
\end{equation}
In this notation, minimization problem (\ref{minprobgen}) takes the form

\begin{equation}
 \text{Minimize }\frac{\prod_{j=1}^{m}n_j}{m!}\text{ over }\{{\bf n}=(n_1,\dots,n_m)\in \mathbb{N}^m: \text{(\ref{L1}) holds and } 1\leq n_1<\dots <n_m\} \label{minprobni}
\end{equation}

For $\mu=1$, Problem (\ref{minprobni}) has a solution only for $m=1$, and we find $h=\delta^{(1)}$, but for $\mu<1$, the problem has a nontrivial solution for all $m$. That is, we can find $n_j$, $j=1,\dots, m$, that satisfy (\ref{L1}). 
In particular, for $n_j(\sigma)=1+\sigma (j-1)$, one easily sees that
\begin{eqnarray}
 \lim\limits_{\sigma \rightarrow\infty} \sum\limits_{j=1}^{m} {\pprod\limits_{i=1}^m}\frac{{n_i(\sigma)}}{\vert {n_i(\sigma)}-{n_j(\sigma)}\vert} 
&=&  1.
\end{eqnarray}
So for every $\mu<1$, ${\bf n}(\sigma)$ satisfies constraint (\ref{L1}) for all $\sigma$ large enough.

Furthermore, any minimizer ${\bf n}$ of problem (\ref{minprobni}) must satisfy $n_1=1$. Indeed, otherwise $n_j>1$ for all $j$ and we can define $\tilde{\bf n}$ by $\tilde{n}_j=n_j-1\geq 1$ for all $j=1,\dots, m$. Calculate
\begin{equation}
  \sum\limits_{j=1}^{m} {\pprod\limits_{i=1}^{m}}\frac{\tilde{n}_i}{\vert \tilde{n}_i-\tilde{n}_j\vert}=
 \sum\limits_{j=1}^{m} {\pprod\limits_{i=1}^{m}}\frac{n_i-1}{\vert n_i-n_j\vert}<
\sum\limits_{j=1}^{m} {\pprod\limits_{i=1}^{m}}\frac{n_i}{\vert n_i-n_j\vert}\leq 
2-\mu
\end{equation}
 and 
\begin{equation}
 \frac{\prod_{j=1}^{m}\tilde{n}_j}{m!} < \frac{\prod_{j=1}^{m}n_j}{m!}.
\end{equation}
So ${\bf n}$ cannot be a minimizer.

Hence, we can fix $n_1\equiv 1$, which reduces problem (\ref{minprobni}) to minimizing
\begin{equation} \label{minimint}
 \eta({\bf{n}}):={\prod_{j=2}\limits^{m}n_j}
\end{equation}
over the set $\{{\bf n}=(n_2,\dots,n_m)\in \mathbb{N}^{m-1} | 1<n_2<\dots <n_m\}$ under the constraint
\begin{equation}\label{L1int}
 \sum\limits_{j=1}^{m} {\pprod\limits_{i=1}^{m}}\frac{n_i}{\vert n_i-n_j\vert}\leq 
\gamma ,
\end{equation}
where again $n_1\equiv 1$.
The factor $m!$ in the denominator has been absorbed into the definition of $\eta$ to simplify the notation. Furthermore, we have set $\gamma = 2-\mu$, as the considerations that follow make sense for arbitrary $\gamma>1$ and not only for $\gamma\leq 2$.

{\em Notational Remark:} All quantities in the derivations below depend on $m$. We will suppress this dependence unless it is relevant in a particular argument.


\section{The relaxed minimization problem for optimal filters}
\label{prob_setup}

The variables $n_j$ correspond to the positions of the nonzero entries in the vector $h$, so they are constrained to positive integer values. We will first consider the {\it relaxed minimization problem} without this constraint; this will eventually enable us to draw conclusions about the original problem.
Thus the variables $n_j\in\N$ will be replaced by relaxed variables $x_j\in\R^+$. Furthermore, it turns out to be convenient to replace the index set $\{1,\dots, m\}$ by $\{0,\dots,{m-1}\}$. 

The relaxed minimization problem is specified as follows: Minimize
\begin{equation} \label{minim1}
 \eta({\bf{x}}):={\prod_{j=1}\limits^{m-1}x_j}
\end{equation}

over the set $D=\left\{{\bf x}\in \mathbb{R}^{m-1} | 1<x_1<x_2<\dots<x_{m-1}\right\}$ under the constraint

\begin{eqnarray}
 f({\bf{x}}):=
\sum\limits_{j=0}^{m-1} \pprod\limits_{ i=0}^{m-1}\frac{x_i}{\vert x_i-x_j\vert}
&\leq& \gamma\label{L1b},
\end{eqnarray}
where $x_0\equiv 1$.

Observe that  $f$ is defined and smooth in the open domain $D$.
The following monotonicity property for
$f$ is important in making inferences from the relaxed to the discrete minimization problem. Let ${\bf r}({\bf x})$ be given by $r_j({\bf x})=\frac{x_j}{x_{j-1}}$, $j=1,\dots,m-1$, and set $F({\bf r}) = f({\bf x})$ for ${\bf x}$ such that ${\bf r}={\bf r}({\bf x})$.
\begin{lemma}\label{monotonicity}
The function $F({\bf r})$ is strictly decreasing in each variable $r_j$.
\end{lemma}
\begin{proof}
A simple calculation shows that
\begin{equation}
 F({\bf r}) = \sum\limits_{j=0}^{m-1} \pprod\limits_{ i=0}^{m-1}\frac{x_i}{\vert x_i-x_j\vert} 
=\sum\limits_{j=0}^{m-1} \prod\limits_{i<j} \frac{1}{r_{i+1} r_{i+2} \cdots r_{j}-1}
\prod\limits_{i>j} \frac{1}{1-\frac{1}{r_{j+1} r_{j+2} \cdots r_{i}}},
\end{equation}
from which the monotonicity is immediate.
\end{proof}
\begin{definition}\label{subord}
If ${\bf x}, {\bf y}\in D$ and $1\leq\frac{y_1}{x_1}\leq\dots\leq \frac{y_{m{-}1}}{x_{m{-}1}}$, we say that ${\bf y}$ is {\em subordinate} to ${\bf x}$. 
\end{definition}
Clearly, ${\bf y}$ is subordinate to ${\bf x}$ if and only if $r_j({\bf x})\leq r_j({\bf y})$ for $j=1,\dots, m-1$, so Lemma~\ref{monotonicity} is equivalent to the following:
\begin{corollary}\label{monotonicity2}
 If ${\bf y}$ is subordinate to ${\bf x}$ and ${\bf x}\neq{\bf y}$, then $f({\bf y})<f({\bf x})$.
\end{corollary}

If ${\bf{x}}$ is a minimizer of the constraint optimization problem (\ref{minim1}), (\ref{L1b}), then $f({\bf x})=\gamma$. Indeed, for a proof by contradiction, assume that ${\bf{x}}$ is a minimizer and $f({\bf{x}})<\gamma$. Then for $t\in [0,1)$, we can define $\tilde x_j(t) = (1-t) x_j+t x_0$. Since $f\circ\tilde{\bf{x}}$ is continuous in
$t$ and 
\begin{equation}
 f(\tilde{\bf{x}}(0))=f({\bf x})<\gamma,
\end{equation}
there exists $t>0$ such that $f\left(\tilde{\bf{x}}(t)\right)<\gamma$. 
However, the function
\begin{equation}
 \eta\left(\tilde{\bf x}(t)\right)= {\prod_{j=0}\limits^{m-1}\left((1-t) x_j+t x_0)\right)}
\end{equation}
is decreasing in $t$, so 
\begin{equation}
 \eta\left(\tilde{\bf x}(t)\right) < \eta\left({\bf x}\right),
\end{equation}
and ${\bf x}$ cannot be a minimizer. Hence we can replace constraint (\ref{L1b}) by the equality
\begin{equation}
  f({\bf{x}})=
\sum\limits_{j=0}^{m-1} \pprod\limits_{ i=0}^{m-1}\frac{x_i}{\vert x_i-x_j\vert}
= \gamma. \label{L1c}
\end{equation}
As we now show, this equation defines a smooth manifold within $D$. It is enough to verify that $\nabla f \neq 0$. 
Note first that
\begin{equation}
  \frac{\partial}{\partial x_k} \frac{x_k}{\left| x_k -x_j\right|} =-x_j \frac{1}{x_k-x_j}\frac{1}{\left| x_k -x_j\right|}.
\end{equation}
 Now calculate for $ j\neq k$ using this fact
 \begin{eqnarray}
     \frac{\partial}{\partial x_k} {\pprod \limits_{ i=0 }^{m-1}}\frac{x_i}{\vert x_i-x_j\vert}
      &=&
      \left({\pprod\limits_{\substack{ i=0 \\ i\neq k}}^{m-1}}
      \frac{x_i}{\vert x_i-x_j\vert}\right) \left(-x_j\frac{1}{x_k-x_j}\frac{1}{\left| x_k -x_j\right|}\right)\\
      &=&-\frac{\eta({\bf x})}{x_k}\frac{(-1)^{j}}{x_k-x_j} b_j ,
 \end{eqnarray}
 where we set from now on 
\begin{equation} \label{bjdef}
 b_j({\bf x})=\pprod\limits_{ i=0 }^{m-1}
      \frac{1}{x_i-x_j}.
\end{equation}
Note that $(-1)^j b_j({\bf x})$ is always positive.

Furthermore, for $j=k$,
   \begin{eqnarray}
     \frac{\partial}{\partial x_k} \pprod \limits_{ i=0}^{m-1}
\frac{x_i}{\vert x_i-x_k\vert}
      &=& - \psum_{ l=0}^{m-1}  \frac{\eta({\bf x})}{x_k}\frac{(-1)^k}{x_k-x_l} b_k({\bf x}).
    \end{eqnarray}
   
   Hence
   \begin{eqnarray}
     \frac{\partial  f}{\partial x_k} 
      &=& -\frac{1}{x_k} \eta({\bf x}) \psum_{j=0}^{m-1}   \frac{1}{x_k-x_j}\left((-1)^k b_k({\bf x})+(-1)^j b_j({\bf x})\right). \label{partialf}
   \end{eqnarray}

For $k=m-1$, all terms in the sum are positive. Hence

\begin{eqnarray}
     \frac{\partial  f}{\partial x_{m-1}} &<& 0
\end{eqnarray}
and so $\{{\bf x}| f({\bf x})=\gamma\}$ is a manifold within $D$.

We now show that the infimum of $\eta$ subject to (\ref{L1c}) is attained in $D$. Let 
$\eta_0=\inf\limits_{{\bf x}\in D, f({\bf x})=\gamma}\eta({\bf x})$ and let ${\bf x}^{(n)}\in D\cap \left\{f=\gamma \right\}$ be chosen such that $\lim\limits_{n\rightarrow\infty} \eta({\bf x}^{(n)})=\eta_0$. As before, we set $x_0^{(n)}\equiv 1$.

We first show that ${\bf x}^{(n)}$ is bounded. Define $M:=\sup\limits_{n\in\N} \eta\left({\bf x}^{(n)}\right)$. Then for each $n$,
\begin{equation}
 \|{\bf x}^{(n)}\|_\infty = |{\bf x}^{(n)}_{m-1}|\leq \eta\left({\bf x}^{(n)}\right)\leq M,
\end{equation}
as, for each $i$, $1\leq{\bf x}^{(n)}_i\leq{\bf x}^{(n)}_{m-1}$. Since $M<\infty$, it follows that ${\bf x}^{(n)}$ is bounded. We conclude that ${\bf x}^{(n)}$ must have a convergent subsequence ${\bf x}^{(n_k)}\rightarrow {\bf x}^{(\infty)}$.

Now ${\bf x}^{(\infty)}$ cannot lie on the boundary of $D$. Indeed, for any $0\leq j\neq k\leq m-1$, we have 

\begin{equation}
 \gamma=f({\bf x}^{(n)})\geq \pprod\limits_{ i=0}^{m-1}
\frac{x^{(n)}_i}{\vert x^{(n)}_i-x^{(n)}_j\vert}
\geq \frac{1}{M^{m-2}|x^{(n)}_{j}-x^{(n)}_{k}|},
\end{equation}
which implies that $|x^{(n)}_{j}-x^{(n)}_{k}|\geq \frac{1}{\gamma M^{m-2}}>0$. It follows that ${\bf x}^{(n)}$ stays away from the boundary of $D$, which implies that ${\bf x}^{(\infty)}\in D$. Thus, problem (\ref{minim1}), (\ref{L1c}) must have at least one minimizer ${\bf{x}}_{min}={\bf x}^{(\infty)}$ in $D$. Note that a priori, there can be more than one minimizer.


As $\{{\bf x}| f({\bf x})=\gamma\}$ is a manifold within $D$, every minimizer ${\bf x}_{min}=(x_1,\dots,x_{m-1})$ of the constrained optimization problem given by (\ref{minim1}) and (\ref{L1c}) solves the associated Lagrange multiplier equations, i.e., there exists $\nu=\nu({\bf x}_{min})\in \mathbb R$ such that
\begin{eqnarray}
  \nu \nabla \eta ({{\bf{x}}_{min}}) + \nabla f({\bf{x}}_{min}) &=& 0\label{lagmult},\\
  f({\bf x}_{min})&=&\gamma.\label{lagmult2}
\end{eqnarray}

Combined with (\ref{partialf}) and the relation 
$\frac{\partial}{\partial y_k} \eta({\bf{y}})=\frac{1}{y_k} \eta({\bf y})$, the Lagrange multiplier equations (\ref{lagmult}), (\ref{lagmult2}) take the explicit form 
\begin{eqnarray}
 \psum_{j=0}^{m-1}   \frac{1}{x_k-x_j}\left((-1)^k b_k({\bf x}_{min})+(-1)^j b_j({\bf x}_{min})\right)&=
&\nu, \label{lagexp}\\
 f({\bf x}_{min}) &=& \gamma \label{lagexp2}
\end{eqnarray}
for $k=1,\dots,m-1$ and $x_0\equiv 0$ as before.

Note that any critical point ${{\bf{x}}_{crit}}$ of the minimization problem for $\eta$ on $D$ solves equations (\ref{lagexp}), (\ref{lagexp2}) for some $\nu$. In the following section, we will show that in fact $\eta$ has a unique critical point in $D$.

\section{Solution of the relaxed minimization problem}
\label{solution_relax}
\begin{theorem}\label{mainthm}
The minimum value of $\eta$ on the manifold $\{f=\gamma\}$ in $D$ is given by
\begin{equation}
 \eta=\eta_{min}=\frac{\sinh(2m\beta)}{(2\sinh\beta)^{2m-1}\cosh\beta} \label{etamin}
\end{equation}
where $\beta=\beta(m,\gamma)$ is the unique positive solution of the equation
\begin{equation}
\frac{\cosh\left((2m-1)\beta\right)}{\cosh \beta}=\gamma.\label{betaform}
\end{equation}
The minimum value $\eta_{min}$ is attained at the unique point ${\bf x}_{min}=(x_1,  \dots, x_{m-1})$, where 
\begin{equation} \label{minset}
 x_j= 1+\frac{1}{2\sinh^2\beta}\left(1+ z_j\right), \ \ \ j=1, \dots, m-1.
 \end{equation}
Here $z_j = cos\left(\frac{m-j}{m}\pi \right)$, $j=1,\dots, m-1$, are the zeros of the Chebyshev polynomial of the second kind of degree $m-1$. 
\end{theorem}

\begin{proof}
The minimization problem (\ref{minim1}), (\ref{L1c}) assumes its minimum in $D$, so there must be at least one critical point ${\bf x}_{crit}=(x_1,\dots,x_{m-1})$  with $1<x_1<\dots<x_{m-1}$. 

To prove uniqueness, we will express the associated Lagrange multiplier problem as a nonlinear matrix equation and then show using a rank argument, which is established by Proposition \ref{rankthm}, that the equation can have only the solution given by (\ref{minset}).

As in (\ref{lagexp}), ${\bf x}_{crit}=(x_1,\dots,x_{m-1})$ must satisfy 
 
\begin{equation}\label{lagmult3}
\psum\limits_{j=0}^{m-1}   \frac{1}{x_k-x_j}\left((-1)^k b_k({\bf x}_{crit})+(-1)^j b_j({\bf x}_{crit})\right)=
\nu({\bf x}_{crit}),
\end{equation}
for $k=1,\dots,m-1$ and, again, $b_j({\bf x}_{crit})=\pprod\limits_{ i=0}^{m-1}
      \frac{1}{x_i-x_j}$. 

In matrix notation, the statement reads 
\begin{equation}\label{mateq}
B({\bf x}_{crit}) {\bf v} = \nu({\bf x}_{crit}) {\bf e},
\end{equation}
where ${\bf e}= (1,1,\dots, 1)^T\in \mathbb{R}^{m-1}$, ${\bf v}= (1, -1, 1,-1,\dots)^T\in \mathbb{R}^{m}$  and the matrix-valued function $B: \R^{m-1}\rightarrow \R^{(m-1)\times m}$ is given by
\begin{equation}
B({\bf y})=\left(\begin{array}{ccccc}
 \frac{b_0({\bf y})}{y_1-y_0} & \psum\limits_{j=0}^{m-1}  
 \frac{b_1({\bf y})}{y_1-y_j} &\frac{b_2({\bf y})}{y_1-y_2}&\cdots& \frac{b_{m-1}({\bf y})}{y_1-y_{m-1}}\\
  \frac{b_0({\bf y})}{y_2-y_0} & \frac{b_1({\bf y})}{y_2-y_1} &\psum\limits_{j=0}^{m-1}
   \frac{b_2({\bf y})}{y_2-y_j} &\cdots& \frac{b_{m-1}({\bf y})}{y_2-y_{m-1}}\\
  \vdots&\vdots&\vdots&\ddots &\vdots\\
 \frac{b_{0}({\bf y})}{y_{m-1}-y_0} &  \frac{b_{1}({\bf y})}{y_{m-1}-y_1} &  \frac{b_{2}({\bf y})}{y_{m-1}-y_2} &\cdots &\psum\limits_{j=0}^{m-1}   \frac{b_{m-1}({\bf y})}{y_{m-1}-y_j}
\end{array}
\right),
\end{equation}
where ${\bf y}=(y_1,\dots,y_{m-1})$ and as before $y_0\equiv 1$.

For given ${\bf y}=(y_1, \dots, y_{m-1})$ let $p_{\bf y}(s)$ be a polynomial such that 
\begin{equation} \label{critp}
 p_{\bf y}'(s)=\prod\limits_{j=1}^{m-1}(s-y_j).
\end{equation}
For definiteness, we normalize $p_{\bf y}(0)=0$.
Let $\Gamma$ be a positively oriented circle in $\mathbb C$ of radius $R$ large enough to enclose all $y_j$'s, including $y_0\equiv 1$.
We now calculate the integral 
\begin{equation}\label{ikintegral}
I_k = \frac{1}{2\pi i}\oint\limits_\Gamma \frac{p_{\bf y}(z)}{(z-y_k)(z-y_0)p_{\bf y}'(z)} dz, \hspace{12pt} k=1,\dots, m-1
\end{equation}
in two different ways. 

Firstly, letting $R\rightarrow \infty$, we see that $I_k = \frac{1}{m}$.
Secondly, we compute the integral using the residues at $y_j$, $0\leq j\leq m-1$.
For the residue $R_j$ at $y_j$, $j\neq k$, we obtain
\begin{equation}
 R_j  = (-1)^{m-1} \frac{b_j({\bf y})}{y_j-y_k}p(y_j).
\end{equation}
At $z_k$, we have a double root in the denominator of the integrand in (\ref{ikintegral}), so
\begin{eqnarray}
R_k = \left.\left(\frac{p_{\bf y}(z)}{\prod\limits_{\substack{i=0\\i\neq k}}^{m-1} (z-y_i)}\right)'\right|_{z=y_k} 
&=& \frac{p_{\bf y}'(y_k)}{\prod\limits_{\substack{i=0\\i\neq k}}^{m-1} (y_k-y_i)}-p_{\bf y}(y_k)\sum\limits_{\substack{j=0\\j\neq k}}^{m-1} \frac{\prod\limits_{\substack{i=0\\i\neq j,k}}^{m-1} (y_k-y_i)}{\left(\prod\limits_{\substack{i=0\\i\neq k}}^{m-1} (y_k-y_i)\right)^2}\\
&=&    (-1)^{m-1}\psum\limits_{j=0}^{m-1} \frac{ b_k({\bf y})}{y_j-y_k} p_{\bf y} (y_k)
\end{eqnarray} 

Summing the residues, we conclude that for $k=1,\dots, m-1$:
\begin{equation}\label{reseq}
\frac{(-1)^{m-1}}{m} = \psum\limits_{j=0}^{m-1}
 \frac{ 1}{y_j-y_k} \left(b_k({\bf y})p_{\bf y}(y_k) +b_j({\bf y}) p_{\bf y}(y_j)\right)
\end{equation}
or equivalently
\begin{equation}\label{mateq2}
B ({\bf y}){\bf p_y} = \frac{(-1)^m}{m}{\bf e},
\end{equation}
where ${\bf p_y} = (p_{\bf y}(y_0), p_{\bf y}(y_1),\dots, p_{\bf y}(y_{m-1})$.

The normalization $p_{\bf y}(0)=0$ plays no role in the above calculation, and so (\ref{mateq2}) also holds for ${\bf p_y}+{\bf e}$. Hence, the vector ${\bf e}$ lies in the kernel of $B(\bf{y})$ for any ${\bf y}$. In Proposition \ref{rankthm}, we will show that $\operatorname{dim}\operatorname{Ker} B({\bf y})=1$, and hence $\operatorname{Ker} B({\bf y})$ is spanned by ${\bf e}$. In particular, this shows that $\nu({\bf x}_{crit})\neq 0$: Otherwise, $v$ would be collinear to $e$, which is impossible.

Specifying ${\bf y}= {\bf x}_{crit}$, one obtains 
\begin{equation}\label{mateq3}
B ({\bf x}_{crit}){\bf p}_{{\bf x}_{crit}} = \frac{(-1)^m}{m}{\bf e},
\end{equation}
and it follows that 
\begin{equation}\label{mateq4}
B ({\bf x}_{crit})\left[m (-1)^m \nu({\bf x}_{crit}) {\bf p}_{{\bf x}_{crit}}\right] = \nu({\bf x}_{crit}){\bf e}.
\end{equation}
By (\ref{mateq}), $\bf v$ also solves (\ref{mateq4}), and thus 
\begin{equation}
 {\bf v} - m (-1)^m \nu({\bf x}_{crit}){\bf p}_{{\bf x}_{crit}}  = c\,{\bf e} 
\end{equation}
for some constant $c$.

Set 
\begin{equation}
 q(s)= m (-1)^m \nu({\bf x}_{crit}) p_{{\bf x}_{crit}} (s)  +c.
\end{equation}
Then $q$ is a polynomial of degree $m$ with critical points at the $x_j$, $j=1,\dots, m-1$ such that $q(x_j)=(-1)^j$, $j=1,\dots, m-1$.
As $q$ cannot have any more critical points and must be monotonic for $x>x_{m-1}$, then ultimately, it will change sign and there is a unique point $x_m>x_{m-1}$ such that $q(x_m)=-q(x_{m-1})=(-1)^m$.
Hence the polynomial $u$ given by 
\begin{equation}
 u(s)=(-1)^m q\left(\frac{x_m-1}{2} (s+1)+1\right)
\end{equation}
has the equi-oscillation property, and we conclude by Proposition \ref{equosc} that $u(s)=T_m(s)$.
That implies, that if $z_m$, $j=1,\dots, m-1$, are the extrema of $T_m$ -- i.e., the zeros of the Chebyshev polynomials of the second kind of degree $m-1$ -- then the extrema of $q$ are given by
\begin{equation}
x_j= \frac{x_m-1}{2}(1+z_j)+1.
\end{equation}
Thus all critical points ${\bf x}_{crit}=(x_1,\dots,x_{m-1})$ are given by
\begin{equation}
 x_j = x_j(K):=1 + K (1 + z_j) \label{xK}
\end{equation}
for some constant $K$.
 It follows from Lemma \ref{monotonicity} that $f({\bf x}(K))$ is strictly monotonic in $K$, i.e., different values of $K$ correspond to different values of $\gamma$. This proves that $\eta$ has a unique critical point on $\{f=\gamma\}$
in $D$, and, in particular, that ${\bf x}_{min}$ is unique and given by (\ref{minset}).
 
 We now compute $K=K(m,\gamma)$. The calculation uses several facts about Chebyshev polynomials and their roots, which we have collected in the Appendix. Let $\beta=\beta(m, \gamma)>0$ be defined through the relation 
 \begin{equation}
 K=\frac{1}{2\sinh^2\beta}.\label{Kbeta}
\end{equation}
 Noting that $z_i=-z_{m-i}$, we obtain from Lemma~\ref{sinprod}
\begin{eqnarray}
 f({\bf x}) & =&  \prod\limits_{i=1}^{m-1} \frac{1+K (1+z_i)}
{K\left\vert z_i-z_0\right\vert}+\sum\limits_{j=1}^{m-1} \prod\limits_{\substack{i=0\\i\neq j}}^{m-1}
\frac{1+K (1+z_i)}
{K \left\vert z_i-z_j\right\vert}\\
&=&\frac{2^{m-1}}{m}\prod\limits_{i=1}^{m-1} \frac{1+K (1+z_i)}
{K}+\sum\limits_{j=1}^{m-1}\frac{2^{m-1}(1-z_j)}{m} \prod\limits_{\substack{i=0\\i\neq j}}^{m-1}
\frac{1+K (1+z_i)}
{K}\\
&=& \left[\frac{2^{m-1}}{m}\prod\limits_{i=1}^{m-1}
\left(\frac{1}{K} +1-z_{m-i}\right) \right]
\left[1+\sum\limits_{j=1}^{m-1}\frac{1+z_{m-j}}
{1+K(1-z_{m-j})} \right]\\
&=& \left[\frac{2^{m-1}}{m}\prod\limits_{i=1}^{m-1}
\left(\frac{1}{K} +1-z_{i}\right) \right]
\left[1+\frac{1}{K}\sum\limits_{j=1}^{m-1}\frac{1+z_{j}}{\left(1+\frac{1}{K}\right)-z_{j}} \right].
\end{eqnarray}
Now  $1+\frac{1}{K}=1+2\sinh^2(\beta)=\cosh(2\beta)$ and
\begin{equation}
\frac{2^{m-1}}{m}\prod\limits_{i=1}^{m-1}
\left(\frac{1}{K} +1-z_{i}\right) = \frac{1}{m^2}T_m'\left(1+\frac{1}{K}\right)=\frac{1}{m^2}T_m'\left(\cosh(2\beta)\right)=\frac{\sinh(2 m \beta)}{m \sinh(2\beta)}.\label{prodfinal}
\end{equation}
Furthermore, differentiating  (\ref{coshcheby}), we obtain for 
$z=\cosh(\tau)$
\begin{equation}
\sum_{j=1}^{m-1}\frac{1}{z-z_j}= \frac{T_m''(z)}{T_m'(z)}= \frac{m\coth(m\tau)-\coth(\tau)}{\sinh{\tau}}.
\end{equation}
Hence
\begin{eqnarray}
& &1+\frac{1}{K}\sum\limits_{j=1}^{m-1}\frac{1+z_{j}}{\left(1+\frac{1}{K}\right)-z_{j}}\\
&=& 1 +(\cosh(2\beta)-1)\sum\limits_{j=1}^{m-1}\left( -1 + \frac{1+\cosh(2\beta)}{\cosh(2\beta)- z_{j}} \right)\\
&=& 1- (m-1)(\cosh(2\beta)-1)+ (\cosh^2(2\beta)-1)\sum\limits_{j=1}^{m-1} \frac{1}{{\cosh(2\beta)- z_{j}}}\\
&=& 1- (m-1)(\cosh(2\beta)-1)
+ \sinh^2({2\beta}) \frac{m\coth(2m\beta)-\coth(2\beta)}{\sinh{2\beta}}\\
&=&m\left( 1- \cosh(2\beta) +  \sinh({2\beta}) {\coth(2m\beta)}\right)\label{sumfinal}.
\end{eqnarray}
Combining (\ref{sumfinal}) and (\ref{prodfinal}) yields
\begin{eqnarray}
 \gamma= f({\bf x}) & =& \frac{\sinh(2 m \beta) - \cosh(2\beta)\sinh(2 m \beta) + \sinh({2\beta})\cosh(2 m \beta) }{ \sinh(2\beta)}\\
&=& \frac{\sinh(2m\beta)-\sinh((2m-2)\beta)}{2\sinh(\beta)\cosh(\beta)}\\
&=& {\frac{2\cosh((2m-1)\beta)\sinh(\beta)}{2\sinh(\beta)\cosh(\beta)}}\\
&=& \frac{\cosh((2m-1)\beta)}{\cosh(\beta)},
\end{eqnarray}
which proves (\ref{betaform}). As $\frac{\cosh((2m-1)\beta)}{\cosh(\beta)}$ is strictly monotonic in $\beta$, $\beta>0$ is uniquely determined from $\gamma$. Of course, this fact also follows from the uniqueness of $K$ proved above.

Now finally, using (\ref{prodfinal}), we write
\begin{equation}
 \eta_{min}
 = \prod\limits_{i=0}^{m-1} (1+ K (1+ z_i)) 
= K^{m-1} \prod\limits_{i=0}^{m-1} \left(\frac{1}{K} +1 +z_i\right)
= \frac{\sinh(2 m \beta)}{(2\sinh(\beta))^{2m-1} \cosh(\beta)}
\end{equation}
\end{proof}

It remains to show that $B({\bf x}_{crit})$ has rank $m-1$. We will show, more generally, that $B({\bf y})$ has rank $m-1$ for an arbitrary ${\bf y}=(y_1,\dots,y_{m-1})$, as long as $y_{i}\neq y_j$ for $i\neq j$ in $\{0,1,\dots, m-1\}$. As before we set $y_0\equiv 1$. The proof of Proposition \ref{rankthm} below goes through without this restriction on $y_0$, but this more general fact is of no consequence for the results in this paper.

Factor out $b_j({\bf y})$ from the $j$-th column, $j=0, \dots {m-1}$ and extend the resulting matrix to an $m\times m$ square matrix $\tilde{B}({\bf y})$ by adding a row that is the negative of the sum of all the other rows, as follows. 
\begin{equation}
\tilde B({\bf y})=\left(\begin{array}{ccccc}
\psum\limits_{l=0}^{m-1}   \frac{1}{y_0-y_l} & \frac{1}{y_0-y_1} & \frac{1}{y_0-y_2}&\cdots& \frac{1}{y_0-y_{m-1}}\\
 \frac{1}{y_1-y_0} & \psum\limits_{l=0}^{m-1} \frac{1}{y_1-y_l} &\frac{1}{y_1-y_2}&\cdots& \frac{1}{y_1-y_{m-1}}\\
  \frac{1}{y_2-y_0} & \frac{1}{y_2-y_1} &\psum\limits_{l=0}^{m-1}   \frac{1}{y_2-y_l} &\cdots& \frac{1}{y_2-y_{m-1}}\\
  \vdots&\vdots&\vdots&\ddots &\vdots\\
 \frac{1}{y_{m-1}-y_0} &  \frac{1}{y_{m-1}-y_1} &  \frac{1}{y_{m-1}-y_2} &\cdots &\psum\limits_{l=0}^{m-1}   \frac{1}{y_{m-1}-y_l} 
\end{array}
\right)
\end{equation}
Clearly, $\operatorname{rank} \tilde B({\bf y})=\operatorname{rank} B({\bf y})$.
We prove that $\tilde B({\bf y})$ has rank $m-1$ by explicitly showing that $\tilde B({\bf y})$ is similar to the Jordan block
 \begin{equation}
  J=\left(\begin{array}{cccc}
0&1&0&\cdots\\
0&0&1&\cdots\\
0&0&0&\cdots\\
\vdots&\vdots&\vdots&\ddots
\end{array}
\right).
 \end{equation}

\begin{proposition}\label{rankthm}
For $m\geq 2$,  $\tilde B({\bf y})$ has the Jordan decomposition
\begin{equation}
\tilde B({\bf y}) = P({\bf y}) J P({\bf y})^{-1}, \label{jordan}
\end{equation}
where 
\begin{equation}
 P({\bf y}) =\left(\begin{array}{cccc}
{b_0({\bf y})}&{b_0({\bf y})} (y_0-y_{m-1})&\cdots&{b_0({\bf y})}\frac{(y_0-y_{m-1})^{m-1}}{(m-1)!}\\
{b_1({\bf y})}&{b_1({\bf y})} (y_1-y_{m-1})&\cdots&b_1({\bf y})\frac{(y_1-y_{m-1})^{m-1}}{(m-1)!}\\
{b_2({\bf y})}&{b_2({\bf y})} (y_2-y_{m-1})&\cdots&b_2({\bf y})\frac{(y_2-y_{m-1})^{m-1}}{(m-1)!}\\
\vdots&\vdots&\ddots&\vdots\\
{b_{m-2}({\bf y})}&{b_{m-2}({\bf y})} (y_{m-2}-y_{m-1})&\cdots&b_{m-2}({\bf y})\frac{(y_{m-2}-y_{m-1})^{m-1}}{(m-1)!}\\
b_{m-1}({\bf y})& 0 & \cdots &0
\end{array}\right).
\end{equation}
Here the $b_j({\bf y})$'s are defined as in (\ref{bjdef}).
\end{proposition}
\begin{proof}
The matrix $P({\bf y})$ is of the form $D_1 V D_2$, where $D_1, D_2$ are invertible diagonal matrices and $V$ is a Vandermonde matrix. Hence, $P({\bf y})$ is invertible and the proof of (\ref{jordan}) is equivalent to showing that $\tilde B({\bf y}) P({\bf y}) = P({\bf y})J$, that is, for $0\leq j,n \leq m-1$,

\begin{equation} \label{rankeq}
\sum\limits_{\substack{k=0\\ k\neq j}}^{m-1}   \frac{b_k({\bf y})}{y_j-y_k}\frac{(y_k-y_{m-1})^{n}}{n!}+
 \sum\limits_{\substack{l=0\\ l\neq j}}^{m-1}   \frac{b_j({\bf y})}{y_j-y_l}\frac{(y_j-y_{m-1})^{n}}{n!}= {b_j({\bf y})}\frac{(y_j-y_{m-1})^{n-1}}{(n-1)!},
\end{equation}
where $\frac{1}{(-1)!}=0$.

 The proof is based on the counterclockwise integral defined for all $t\in\C$
\begin{equation}
J_{m,n}(t)=\frac{1}{2\pi i}\oint\limits_\Gamma \frac{(z-t)^n}{\prod\limits_{i=0}^{m-1}(z-y_i)} dz, \ \ \ \ 0\leq n\leq m-1
\end{equation}
over a circle $\Gamma$ of radius $R$ large enough that it encloses all $y_j$'s. Letting $R\rightarrow \infty$, we see that $J_{m,n}=\delta^{(0)}_{n-(m-1)}$ independent of $t$. On the other hand, note that the residue at $y_k$ is $(-1)^{m-1}b_k({\bf y}) (y_k-t)^n$. Hence
\begin{equation}\label{indseed1}
\delta^{(0)}_{n-(m-1)}=J_{m,n}=(-1)^{m-1} \sum\limits_{k=0}^{m-1} b_k({\bf y})(y_k-t)^n.
\end{equation}
Now
\begin{equation}
 \frac{\partial b_k}{\partial y_j} =
\begin{cases}
  \frac{b_k({\bf y})}{y_k-y_j}\ \ \ \ \ \ \ \ \ \ \ \  \text{for $j\neq k$}\\
\ \\
\sum\limits_{\substack{l=0\\l\neq j}}^{m-1}   \frac{b_j({\bf y})}{y_l-y_j} \ \ \ \ \ \ \  \text{for $j= k$}
\end{cases}
\end{equation}
Hence, differentiating (\ref{indseed1}) with respect to $y_j$, leads to the identity 
\begin{equation}
 \sum\limits_{\substack{k=0\\k\neq j}}^{m-1}   \frac{b_k({\bf y})}{y_k-y_j}(y_k-t)^n+
 \sum\limits_{\substack{l=0\\l\neq j}}^{m-1}   \frac{b_j({\bf y})}{y_l-y_j}(y_l-t)^n + 
 b_l({\bf y})n (y_j-t)^{n-1} =0\label{casen0}.
 \end{equation}
Letting $t\rightarrow y_{m-1}$, one obtains (\ref{rankeq}).
\end{proof}

\section{Asymptotics for the relaxed and the discrete minimization problem\label{asymp}}

In the following proposition, we evaluate the dependence on $m$ of the solution ${\bf{x}}={\bf{x}}^{(m)}$ of the relaxed minimization problem. For any fixed $j$, we show that $x_j^{(m)}$ converges as $m\rightarrow\infty$, and we compute the limit.
\begin{proposition}\label{propasymp}
 \begin{enumerate}  
 \item[\rm{(a)}] For $K=K(m,\gamma)$ as in (\ref{xK})
 \begin{equation}
  \frac{2(m-1)^2}{(\cosh^{-1}\gamma)^2}-1\leq K \leq \frac{2m^2}{(\cosh^{-1}\gamma)^2}. \label{parta}
 \end{equation}
 \item[\rm{(b)}] Set $\sigma:=\frac{\pi^2}{\left(\cosh^{-1}\gamma\right)^2}$. Then for all $m$ and all $1\leq j\leq m-1$,
 \begin{equation}
 x^{(m)}_j\leq 1+\sigma j^2.
 \end{equation}
 \item[\rm{(c)}] 
 For any fixed $j\geq 1$,
 \begin{equation}
  \lim_{m\rightarrow\infty} x_j^{(m)} =  1+\sigma  j^2.
 \end{equation}
  \item[\rm{(d)}]
  \begin{equation}
  \lim_{m\rightarrow\infty} \frac{\left(\eta({\bf x}^{(m)})\right)^{1/m}}{m^2}= \frac{1}{(\cosh^{-1}\gamma)^2} = \frac{\sigma}{\pi^2}.
  \end{equation}
 \end{enumerate}
\end{proposition}
\begin{proof}
We first provide bounds on $\beta$ defined in (\ref{betaform}).
For a lower bound, write
\begin{equation}
\gamma= \frac{\cosh(2m-1)\beta)}{\cosh\beta} = \cosh(2m\beta)- \sinh(2m\beta)\tanh\beta \leq \cosh(2m\beta).
\end{equation}
For an upper bound, we have
\begin{equation}
\gamma= \frac{\cosh(2m-1)\beta)}{\cosh\beta} =\cosh((2m-2)\beta) + \sinh((2m-2)\beta)\tanh\beta\geq \cosh((2m-2)\beta).
\end{equation}
We obtain the bounds
\begin{equation}
\frac{1}{2m}\cosh^{-1}\gamma \leq \beta \leq\frac{1}{2m-2}\cosh^{-1}\gamma.\label{betabounds}
\end{equation}
This implies the upper bound for $K$
\begin{equation}
K=\frac{1}{2\sinh^2\beta}\leq\frac{1}{2\beta^2}\leq\frac{2m^2}{(\cosh^{-1}\gamma)^2}.
\end{equation}
For the lower bound on $K$, we have by an elementary estimate
\begin{equation}
 K=\frac{1}{2\sinh^2\beta}\geq\frac{1}{2\beta^2}-1\geq \frac{2(m-1)^2}{(\cosh^{-1}\gamma)^2}-1.
\end{equation}
This proves (a).

Also 
\begin{equation}
  x^{(m)}_j = 1 + 2 K \sin^2\left(\frac{j\pi}{2m}\right) \leq 1+ \frac{4m^2}{(\cosh^{-1}\gamma)^2}\left(\frac{j\pi}{2m}\right)^2 =1 +\sigma j^2,
\end{equation}
which proves (b).

From (\ref{parta})
\begin{equation}
\lim\limits_{m\rightarrow\infty} \frac{K(m,\gamma)}{m^2}=\frac{2}{(\cosh^{-1}\gamma)^2},
\end{equation}
and so
\begin{equation}
\lim\limits_{m\rightarrow\infty} x_j^{(m)} =1+\frac{2}{(\cosh^{-1}\gamma)^2}\lim\limits_{m\rightarrow\infty} 2 m^2 \sin^2\left(\frac{j\pi}{2m}\right) =   1 + \sigma j^2,
\end{equation}
which proves (c).

Finally, from (\ref{betabounds}) we see that
\begin{equation}
\lim_{m\rightarrow\infty} 2m\beta(m,\gamma) = \cosh^{-1} \gamma,
\end{equation}
and hence from (\ref{etamin})
  \begin{eqnarray}
  \lim_{m\rightarrow\infty} \frac{\left(\eta({\bf x})\right)^{1/m}}{m^2}&=&\lim_{m\rightarrow\infty} \frac{1}{m^2}\left (\frac{\sinh(2m\beta)}{(2\sinh\beta)^{2m-1}\cosh\beta}\right)^{1/m} \\
  &=&  \lim\limits_{m\rightarrow\infty}\left(\frac{1}{4 m^2 \sinh^2\beta}\right)\left(\frac{2\sin\beta \sinh\left(2m\beta\right)}{\cosh\beta}\right)^{1/m}\\
  &=& \lim\limits_{m\rightarrow\infty}\frac{K(m,\gamma)}{2 m^2} = \frac{1}{(\cosh^{-1}\gamma)^2},
  \end{eqnarray}
which proves (d). This completes the proof of the Proposition.
\end{proof}

Motivated by the above proposition, let us define 
\begin{equation}\label{def-w}
w_j := 1+\sigma j^2, ~~~~~j\geq 1\dots,
\end{equation} 
and an associated vector sequence
${\bf w}^{(m)} := (w_1, \dots, w_{m-1})$, $m\geq 2$.
We know that for each $j\geq 1$, 
$x^{(m)}_j \to w_j$ as $m\to \infty$. The following
lemma provides a nonasymptotic relation between 
${\bf x}^{(m)}$ and ${\bf w}^{(m)}$ which will be utilized in proving
Proposition \ref{propasymp-n} and Theorem \ref{thmasop}.

\begin{lemma}\label{ljsubor}
${\bf w}^{(m)}$ is subordinate to ${\bf x}^{(m)}$.
\end{lemma}

\begin{proof}
By Definition~\ref{subord}, we need to show that for $0\leq j\leq m-2$ and $w^{(m)}_0\equiv x_0^{(m)}\equiv 1$,
\begin{equation}\label{w-x-ratio}
 \frac{w_{j+1}^{(m)}}{x_{j+1}^{(m)}}\geq \frac{w_{j}^{(m)}}{x_{j}^{(m)}},
\end{equation}
that is,
\begin{equation}
 (1+\sigma (j+1)^2)\left(1+2K\sin^2 \left(\frac{j\pi}{2m}\right)\right) \geq
 (1+\sigma j^2)\left(1+2K\sin^2 \left(\frac{(j+1)\pi}{2m}\right)\right),
\end{equation}
or equivalently
\begin{eqnarray}
 \left[\sigma(2j+1) - 2 K \left(\sin^2 \left(\frac{(j+1)\pi}{2m}\right)
-\sin^2 \left(\frac{j\pi}{2m}\right)\right)\right] & & \nonumber \\
+ \left[2\sigma K \left( (j+1)^2\sin^2 \left(\frac{j\pi}{2m}\right)- j^2 \sin^2 \left(\frac{(j+1)\pi}{2m}\right)\right)\right]&\geq& 0. \label{lsuboxexpand}
\end{eqnarray}
We show that both these summands are nonnegative. 

By Proposition \ref{propasymp}(a), 
$K\leq \frac{2m^2\sigma}{\pi^2}$, and so for the first summand, it is sufficient to show that
\begin{equation}
(2j+1) - \frac{4m^2}{\pi^2}\left(\sin^2 \left(\frac{(j+1)\pi}{2m}\right) -\sin^2 \left(\frac{j\pi}{2m}\right)\right)
\geq 0\label{Kmtos}.
\end{equation}
Indeed, by standard trigonometric identities
\begin{equation}
\frac{4m^2}{\pi^2}\left(\sin^2 \left(
\frac{(j+1)\pi}{2m}\right) -\sin^2 \left(\frac{j\pi}{2m}\right)\right)
=\frac{4m^2}{\pi^2} \sin \left(
\frac{(2j+1)\pi}{2m}\right)\sin\left(\frac{\pi}{2m}\right)
\leq 2j+1,
\end{equation}
which proves (\ref{Kmtos}).

On the other hand, the positivity of the second summand follows from the fact that the function $\frac{\sin(y)}{y}$ is decreasing on $[0,\frac{\pi}{2}]$. This completes the proof of (\ref{lsuboxexpand}) and hence the proof of the Lemma.
\end{proof}

The above results for the relaxed minimization problem allow us to draw conclusions for our original problem with the constraint that the filter locations $n^{(m)}_j$ are all integers. 

With $n^{(m)}_1\equiv x^{(m)}_0\equiv 1$, we seek an integer sequence ${\bf n}^{(m)}= (n_2^{(m)},\dots,n^{(m)}_m)$ such that ${\bf n}^{(m)}$ is subordinate in the sense of Definition~\ref{subord} to  
${\bf x}^{(m)}:=\left(1+K(1+z_j) \right)_{j=1}^{m-1}$, the solution of the relaxed minimization problem (\ref{minim1}), (\ref{L1b}).  By Corollary~\ref{monotonicity2},  $f({\bf{n}}^{(m)})\leq f({\bf{x}}^{(m)})\leq \gamma$, so ${\bf n}^{(m)}$ satisfies (\ref{L1int}). 
Note that for the $n^{(m)}_j$'s we use the original index set $j=1, \dots, m$ of Section~\ref{minsupp}, while for the $x^{(m)}_j$'s we retain the labels $j=0, \dots, m-1$. 
In this section, we work with the specific integer sequence ${\bf n}^{(m)}=(n_2^{(m)}, \dots, n_m^{(m)})$ defined recursively by
\begin{equation}\label{minsubo}
n_{j+1}^{(m)}= \left\lceil n_{j}^{(m)}\frac{x^{(m)}_{j}}{x^{(m)}_{j-1}} \right\rceil, \ \ \ \ j=1,\dots, m-1,
\end{equation}
where $n^{(m)}_1\equiv x^{(m)}_0\equiv 1$ as above and $\lceil s\rceil$ denotes the smallest integer greater or equal to $s$.
This sequence is minimal amongst all integer sequences subordinate to ${\bf x}^{(m)}$ in the sense that
 if ${\bf k}=(k_2, \dots, k_m)$ is any  integer sequence such that $1\leq\frac{k_2}{x^{(m)}_1}\leq \dots \leq \frac{k_m}{x^{(m)}_{m-1}}$,
 then $k_{j}\geq n^{(m)}_j$ for all $j=2,\dots, m$. Indeed one has $k_2\geq x^{(m)}_1$, which implies $k_2\geq \lceil x^{(m)}_1\rceil=n^{(m)}_2$, and assuming by induction $k_j\geq n^{(m)}_j$, one obtains
\begin{equation}
 k_{j+1} =\left\lceil k_{j+1}\right\rceil \geq \left\lceil x^{(m)}_{j}\frac{k_{j}}{x^{(m)}_{j-1}} \right\rceil\geq \left\lceil x^{(m)}_{j}\frac{n^{(m)}_{j}}{x^{(m)}_{j-1}} \right\rceil=n^{(m)}_{j+1}.
\end{equation}

We will next derive an asymptotic formula, analogous to Proposition 
\ref{propasymp}(c), for the limit of $n^{(m)}_j$ as $m \to \infty$.

\begin{proposition} \label{propasymp-n}
Let $\sigma > \frac{5}{4}$. Then, for all $j \geq 1$,
\begin{equation} \label{lim-nj}
\lim_{m \to \infty} n^{(m)}_j = 1 + \lceil \sigma \rceil (j-1)^2.
\end{equation}
\end{proposition}

\begin{proof}

We will prove the statement by induction on $j$. The case $j=1$ holds by definition. Assume that \eqref{lim-nj} holds for some $j \geq 1$.

We first consider the case when $\sigma$ is an integer. 
Since $n^{(m)}_j$ is integer valued and $\lceil \sigma \rceil = \sigma$, \eqref{lim-nj} is equivalent to saying that
$n_j^{(m)} = 1 + \sigma (j-1)^2$  for all sufficiently large $m$. 
Now by Lemma \ref{ljsubor} 
(see \eqref{w-x-ratio}) we have 
\begin{equation}
1 + \sigma j^2 \geq (1+\sigma(j-1)^2) \frac{x_j^{(m)}}{x_{j-1}^{(m)}}
= n_j^{(m)} \frac{x_j^{(m)}}{x_{j-1}^{(m)}}
\end{equation}
which, together with \eqref{minsubo}, implies
$1 + \sigma j^2 \geq n_{j+1}^{(m)}$. Since $n_{j+1}^{(m)} \geq x_j^{(m)}$ and $x_j^{(m)} \to 1 + \sigma j^2$ (see
Proposition~\ref{propasymp}(c)), we obtain
$n_{j+1}^{(m)} \to 1 + \sigma j^2$, which completes the induction step.

Now assume $\sigma > \frac{5}{4}$ is noninteger.
 By Proposition~\ref{propasymp}(c) and the induction hypothesis, we have
\begin{equation} \label{lim-exp1}
\lim_{m \to \infty} n_{j}^{(m)}\frac{x^{(m)}_{j}}{x^{(m)}_{j-1}}
= (1+\lceil \sigma \rceil (j-1)^2) \frac{1+\sigma j^2}{1+\sigma (j-1)^2}
\end{equation}
Since the function $\lceil \cdot \rceil$ is continuous at every
noninteger, our problem reduces to showing that
\begin{equation} \label{equiv-exp}
\lceil \sigma \rceil j^2 < 
(1+\lceil \sigma \rceil (j-1)^2) \frac{1+\sigma j^2}{1+\sigma (j-1)^2}
< 1 + \lceil \sigma \rceil j^2
\end{equation}
for all $j \geq 1$ and $\sigma > \frac{5}{4}$, for then we have
\begin{equation}
1 + \lceil \sigma \rceil j^2 =
\left \lceil \lim_{m \to \infty} n_{j}^{(m)}\frac{x^{(m)}_{j}}{x^{(m)}_{j-1}}
\right \rceil
= \lim_{m \to \infty} 
\left \lceil n_{j}^{(m)}\frac{x^{(m)}_{j}}{x^{(m)}_{j-1}}
\right \rceil
= \lim_{m \to \infty} n_{j+1}^{(m)},
\end{equation}
which would complete the induction step. 

To show \eqref{equiv-exp}, first note that
\begin{equation} \label{algeb-iden}
(1+\lceil \sigma \rceil (j-1)^2) \frac{1+\sigma j^2}{1+\sigma (j-1)^2}
- \lceil \sigma \rceil j^2 = 1 - 
\frac{(\lceil \sigma \rceil - \sigma)(2j-1)}{1+\sigma (j-1)^2},
\end{equation}
which immediately yields the second inequality for all 
$j \geq 1$ and noninteger $\sigma > 0$. The first inequality in
\eqref{equiv-exp} is also immediate for $j=1$, as the right hand side of
\eqref{algeb-iden} reduces to $1 - (\lceil \sigma \rceil - \sigma)$,
which is strictly positive. 
Now note that 
\begin{equation}
1+\sigma (t-1)^2 \geq 2\sigma t + 1-3\sigma,\mbox{ for all }t \in \mathbb{R},
\end{equation}
since the right hand side is the equation of the line tangent to the
parabola $1 + \sigma(t-1)^2$ at $t=2$. Now,
for $\sigma > 2$, we have
\begin{equation}
\frac{(\lceil \sigma \rceil - \sigma) (2j-1)}{1+\sigma (j-1)^2}
< \frac{2j-1}{1+2 (j-1)^2} \leq \frac{2j-1}{4j-5} \leq 1, 
\mbox{ for all } j \geq 2.
\end{equation}
On the other hand, for $\frac{5}{4} < \sigma < 2$, we have
\begin{equation}
\frac{(\lceil \sigma \rceil - \sigma) (2j-1)}{1+\sigma (j-1)^2}
< \frac{\frac{3}{4}(2j-1)}{1+ \frac{5}{4}(j-1)^2} 
\leq \frac{\frac{3}{2}j-\frac{3}{4}}{\frac{5}{2}j-\frac{11}{4}}
\leq 1, \mbox{ for all }
j \geq 2.
\end{equation}
Hence the first inequality in \eqref{equiv-exp} holds for
all $j \geq 1$ and noninteger $\sigma > \frac{5}{4}$. 
\end{proof}

{\em Remark:}
For $j=1$ and $j=2$, we have $n^{(m)}_1 = 1$ and 
$n^{(m)}_2 = \lceil x^{(m)}_1 \rceil$ so that 
the formula \eqref{lim-nj} is actually valid for all 
$\sigma > 0$  because of Proposition \ref{propasymp}(b) and 
(c). On the other hand, for $j=3$ and $1 < \sigma 
< \frac{5}{4}$, we have $n^{(m)}_3 \to 8$ instead of $9$.
To see this, note that $n^{(m)}_2 \to 1 + \lceil \sigma \rceil = 3$, so that
\begin{equation} \label{lim-exp2}
\lim_{m \to \infty} n_{2}^{(m)}\frac{x^{(m)}_{2}}{x^{(m)}_{1}}
= \frac{3(1+4\sigma)}{1+\sigma} \in (7.5,8).
\end{equation}
Similarly, for $1 < \sigma < \frac{5}{4}$,
$n^{(m)}_4 \to 17$ instead of $19$. As shown in Figure~\ref{quant_sigma}, the pattern becomes
more complicated for higher values of $j$ and the interval has to be 
subdivided. A similar pattern is present also for $0 < \sigma < 1$.
We shall not explore this here; we again refer to Figure~\ref{quant_sigma}.

\begin{figure}[htbp] 
\centerline{\includegraphics[width=5.5in]{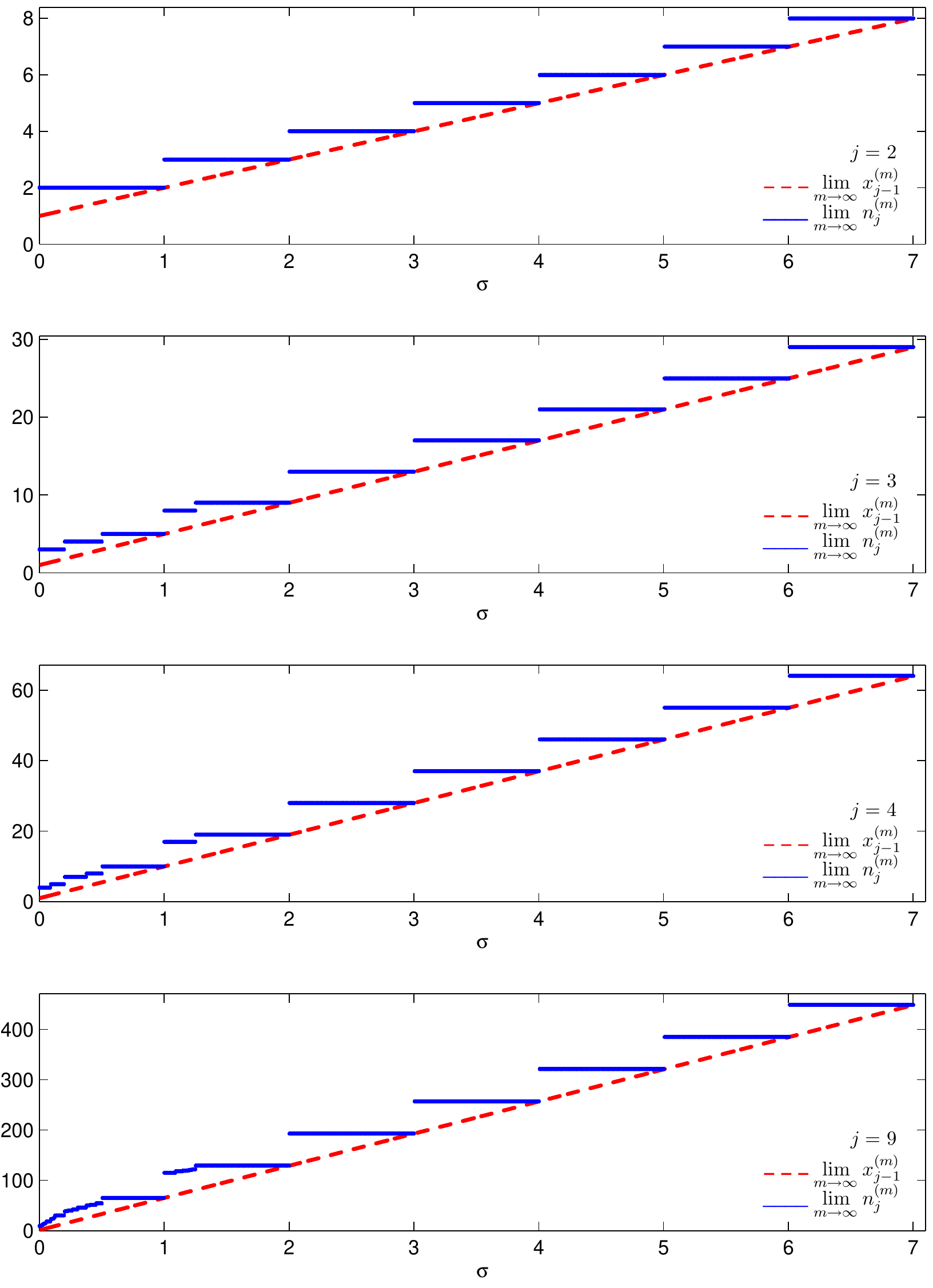}}
\caption{Numerical comparison of $\lim\limits_{m \to \infty} x_{j-1}^{(m)}$ 
with $\lim\limits_{m \to \infty} n_{j}^{(m)}$ as a function of 
$\sigma$ for $j=2,3,4,9$. For $j=2$,
these limits equal $1+\sigma$ and $1+\lceil \sigma \rceil$, respectively.
For $j\geq 3$, notice the difference between the cases $\sigma > 5/4$ and
$\sigma < 5/4$. In the first case, one has 
$\lim\limits_{m \to \infty} n_{j}^{(m)}=1+\lceil \sigma \rceil (j{-}1)^2$, whereas
in the latter case this formula is no longer valid.}
\label{quant_sigma}
\end{figure}

\begin{definition}
 A sequence of integer vectors ${\bf{k}}^{(m)}=(k_2^{(m)},\dots,k_m^{(m)})$ of increasing length $m-1$, $m=2, 3\dots$, with $f(k^{(m)})\leq \gamma$ is said to be {\em asymptotically optimal} if 
\begin{equation}
\lim_{m\rightarrow\infty}\left(\frac{\eta\left({\bf{k}}^{(m)}\right)}{\eta\left({\bf{x}}^{(m)}\right)}\right)^{1/m} =1,
\end{equation}
where ${\bf{x}}^{(m)}$ is the solution of (\ref{minim1}), (\ref{L1b}) as above. 
\end{definition}

The relevance of this definition lies in the fact that the quantity
$\lim\limits_{m \to \infty} \frac{1}{m^2} (\eta({\bf k}^{(m)}))^{1/m}$
controls the rate of exponential decay for the error bound associated
with the sequence of filters defined by $\big({\bf k}^{(m)}\big)$. The precise
relation will be seen in Theorem~\ref{decayrate} below.
We will use the following lemma to assess the asymptotic optimality of 
our minimal subordinate 
construction ${\bf n}^{(m)}$ defined in \eqref{minsubo}.

\begin{theorem}\label{thmasop}
If $\sigma$, defined in Proposition~\ref{propasymp}, is an integer, then the sequence ${\bf n}^{(m)}$, $m=2, 3, \dots$, defined by (\ref{minsubo}), is both asymptotically optimal and  subordinate to ${\bf{x}}^{(m)}$. If $\sigma$ is not an integer, no sequence of integer vectors can have both of 
these properties.
\end{theorem}
\begin{proof}
If $\sigma$ is not an integer, then $\lim\limits_{m\rightarrow\infty} x_1^{(m)}=1+\sigma$ is  not an integer, and so for any sequence ${\bf k}^{(m)}=(k_2^{(m)},\dots,k_m^{(m)})$, $m=2,3,\dots$, of integer vectors subordinate to ${\bf x}^{(m)}$,
\begin{equation}\label{lim-ratio-non-int}
 \limsup\limits_{m\rightarrow\infty}\left(\frac{\eta({\bf k})}{\eta({\bf x})}\right)^{1/m}\geq \limsup\limits_{m\rightarrow\infty}\frac{k_2^{(m)}}{x_1^{(m)}}\geq \lim\limits_{m\rightarrow\infty}\frac{\lceil x_1^{(m)}\rceil}{{x_1^{(m)}}}=\frac{1+\lceil \sigma \rceil}{1+\sigma}>1.
\end{equation}
Hence ${\bf k}^{(m)}$ cannot be asymptotically optimal.

Now consider the case that $\sigma$ is an integer. Then by Lemma~\ref{ljsubor}, ${\bf w}^{(m)}= (w_1, \dots, w_{m-1})$ 
is an integer sequence subordinate to ${\bf x}^{(m)}$. As in Proposition \ref{propasymp-n}, 
one has $x_j^{(m)}\leq n_{j+1}^{(m)}\leq w_j$, as ${\bf n}^{(m)}$ is the minimal integer sequence subordinate to ${\bf x}^{(m)}$.

Next, from Proposition~\ref{propasymp}(a) we see that
\begin{equation}
\frac{2K}{m^2}\geq \frac{4s}{\pi^2}\left(1-\frac{C_1}{m}\right) \label{Klowb}
\end{equation}
for some constant $C_1<\infty$.
Together with the elementary fact that $\left(\frac{\sin x}{x}\right)^2 \geq 1-C_2 x^2$ for a sufficiently large constant $C_2$, this implies that for  $1\leq j \leq m^{2/3}$ and some constant $C_3 < \infty$
\begin{equation}
x_j^{(m)}=1+2 K \sin^2\left(\frac{j\pi}{2m}\right) \geq 1 + \sigma\left(1-\frac{C_1}{m}\right) \left(j^2 \left(1-C_2\frac{\pi^2}{4}\frac{j^2}{m^2}\right)\right)\geq (1+\sigma j^2)\left(1-\frac{C_3}{m^{2/3}}\right).
\end{equation}
Then for $1\leq j \leq m^{2/3}$ and some $C_4<\infty$
\begin{equation}
\frac{{n_{j+1}}^{(m)}}{{x_j}^{(m)}}\leq\frac{{w_j}}{{x_j}^{(m)}}
\leq \frac{1}{1-\frac{C_3}{m^{2/3}}}\leq 1+ \frac{C_4}{m^{2/3}}.
\end{equation}
Now
\begin{equation}
\frac{{n_{j+1}}^{(m)}}{{x_j}^{(m)}}=\frac{1}{{x_j}^{(m)}}\left\lceil n_{j}^{(m)}\frac{x^{(m)}_{j}}{x^{(m)}_{j-1}} \right\rceil\leq \frac{1}{{x_j}^{(m)}}  \left(n_{j}^{(m)}\frac{x^{(m)}_{j}}{x^{(m)}_{j-1}} +1\right) = \frac{n^{(m)}_{j}}{x^{(m)}_{j-1}} + \frac{1}{{x_j}^{(m)}}.\label{njxjrec}
\end{equation}
Combining (\ref{Klowb}) together with the elementary lower bound $\frac{\sin x}{x}\geq \frac{2}{\pi}$ for $0\leq x \leq \frac{\pi}{2}$, we obtain
$ x_j^{(m)}\geq C_5 j^2$, $j\geq 1$, for some constant $C_5>0$.
By repeated application of (\ref{njxjrec}), one then obtains for $m^{2/3}<j\leq m-1$ and some constant $C_6<\infty$
\begin{equation}
\frac{n_{j+1}^{(m)}}{x_j^{(m)}}\leq \frac{n^{(m)}_{j}}{x^{(m)}_{j-1}} + \frac{1}{C_5 j^2}\leq \dots
\leq \frac{n^{(m)}_{\lfloor m^{2/3}\rfloor}+1}{x^{(m)}_{\lfloor m^{2/3}\rfloor}} +\sum\limits_{l=\lfloor m^{2/3}\rfloor+1}^j \frac{1}{C_5 l^2}\leq 1+ \frac{C_4}{m^{2/3}} +\frac{C_6}{m^{2/3}}.
\end{equation}
Thus  there exists some constant $C_7<\infty$, such that for all $1\leq j\leq m-1$, 
\begin{equation}
\frac{{n_{j+1}}^{(m)}}{{x_j}^{(m)}}\leq 1 +\frac{C_7}{m^{2/3}}.
\end{equation}
We conclude that 
\begin{equation}
1\leq \frac{\eta({\bf n}^{(m)})}{\eta({\bf x}^{(m)})} 
= \prod\limits_{j=1}^{m-1}\frac{n_{j+1}^{(m)}}{{x_j}^{(m)}}
\leq \left(1 +\frac{C_7}{m^{2/3}}\right)^{m},
\end{equation}
which implies that 
\begin{equation}
\lim\limits_{m\rightarrow\infty}  \left(\frac{\eta({\bf n}^{(m)})}{\eta({\bf x}^{(m)})}\right)^{1/m} = 1,
\end{equation}
and hence ${\bf n}^{(m)}$ is asymptotically optimal. 
\end{proof}

{\em Remark:} For noninteger values of $\sigma$,
we do not know if there are asymptotically optimal 
integer vectors ${\bf k}^{(m)}$ (which are necessarily not 
subordinate to ${\bf x}^{(m)}$) that are 
admissible, i.e., that satisfy $f({\bf k}^{(m)}) \leq \gamma
= \cosh (\pi/\sqrt{\sigma})$. 
At the same time, it is natural to ask how close our 
construction ${\bf n}^{(m)}$ is to being
asymptotically optimal, i.e., the value of the
limit, as $m\to \infty$, of  $ (\eta({\bf n}^{(m)})/\eta({\bf x}^{(m)}))^{1/m}$. 
We shall not carry out 
an analysis here that is similar to Proposition \ref{propasymp-n}
to find this limit, 
but instead provide our numerical findings in Figure \ref{comparison}.\\

\begin{figure}[htb] 
\centerline{\includegraphics[width=5in]{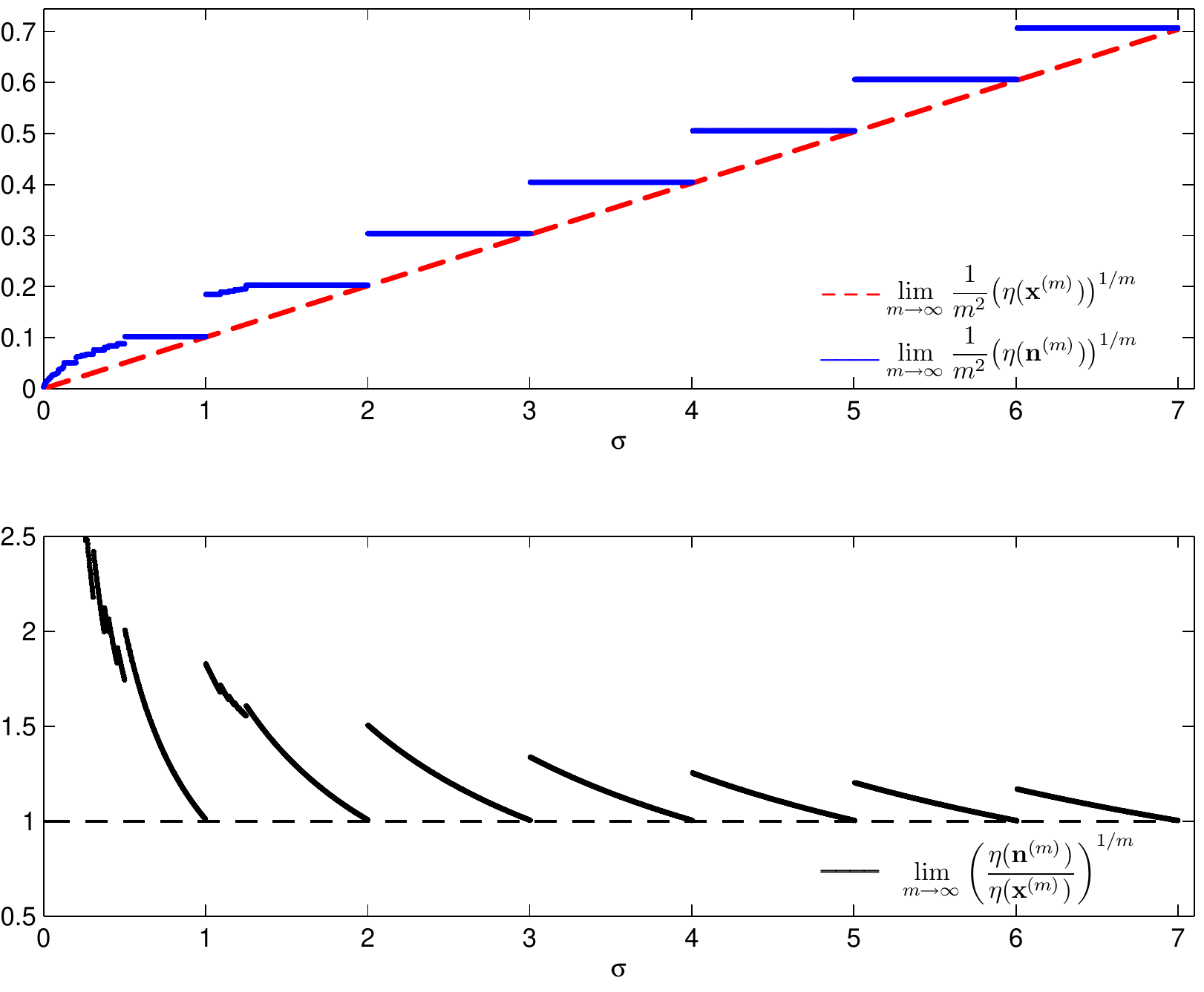}}
\caption{Numerical comparison of $\lim\limits_{m\to \infty} \frac{1}{m^2} 
\big(\eta({\bf n}^{(m)})\big)^{1/m}$ with $\lim\limits_{m\to \infty} 
\frac{1}{m^2} \big(\eta({\bf x}^{(m)})\big)^{1/m} = \sigma/\pi^2$, as a function of $\sigma$.
Top: plotted individually, bottom: their ratio.}
\label{comparison}
\end{figure}

\subsection*{Optimal exponential error decay for minimally supported
filters}

\par We are now ready to prove the promised improved
exponential error decay
estimate for the $\Sigma\Delta$ modulators defined by ${\bf n}^{(m)}$.

\begin{theorem}\label{decayrate}
For all $1<\gamma<2$ such that $\sigma=\frac{\pi^2}{\left(\cosh^{-1} \gamma\right)^2}$ is an integer, all one-bit $\Sigma\Delta$ modulators corresponding to filters $h^{(m)}$ minimally supported at positions $1, n_2^{(m)}, \dots, n_m^{(m)}$ are stable for all input sequences $y$ with $\|y\|_\infty\leq\mu = 2-\gamma$. Furthermore, the family consisting of the one-bit $\Sigma\Delta$ modulators corresponding to the filters $\left\{h^{(m)}\right\}_{m=2}^\infty$ for all orders $m$ gives rise to exponential error decay: For any rate constant $r<r_0:=\frac{\pi}{e^2 \sigma \ln 2}$, there exists a constant $C=C(r)$ such that 
\begin{equation}
\|e_\lambda\|_\infty\leq C 2^{-r\lambda}.
\end{equation}
\end{theorem}
\begin{proof}
Stability follows from the fact that ${\bf n}^{(m)}$ is subordinate to ${\bf x}^{(m)}$, which satisfies the stability condition $f({\bf x}^{(m)})\leq \gamma$. 

Choose the reconstruction kernel $\varphi_0$ such that the corresponding $\epsilon$ as introduced in Section~\ref{basicerror} satisfies $1+\epsilon < \sqrt{\frac{r_0}{r}}$. 
Now let $g^{(m)}$ be such that $\Delta^m g^{(m)} =\delta^{(0)}- h^{(m)}$, as in Section~\ref{optimal}. Then from (\ref{err2}) and (\ref{uvg}), we have the error bound
\begin{equation}\label{initial-error-1}
\|e_\lambda\|_\infty \leq \|g^{(m)}\|_1 \|v\|_\infty \|\varphi_0\|_1 \pi^ m (1+\epsilon)^m \lambda^{-m},
\end{equation}
where $v$ solves (\ref{rec_eq}).

Recall that our construction yields $\|v\|_\infty\leq 1$.
Furthermore, by Theorem~\ref{thmasop} and Proposition~\ref{propasymp}~(d), we have that
   \begin{equation}\label{lim-eta}
 \lim_{m\rightarrow\infty} \frac{\left(\eta({\bf n}^{(m)})\right)^{1/m}}{m^2}= \frac{1}{(\cosh^{-1}\gamma)^2}
= \frac{\sigma}{\pi^2}
  \end{equation}
 and hence by (\ref{geta})
 \begin{equation}
 \|g^{(m)}\|_1 = \frac{\eta({\bf n}^{(m)})}{m!} = \left(\frac{e\sigma}{\pi^2}\right)^m m^m \left(1+o(1)\right)^m.
 \end{equation}

Now consider $m\geq M(r)$ large enough to ensure that the $(1+o(1))$-factor is less than $\sqrt{\frac{r_0}{r}}$. Then
\eqref{initial-error-1} implies
\begin{equation}
\|e_\lambda\|_\infty \leq  \|\varphi_0\|_1  \left(\frac{ e\sigma}{\pi}\right)^m m^m \left(\frac{r_0}{r}\right)^m 
\lambda^{-m}.
\end{equation}
As explained in Section~\ref{optimal}, we choose, for each $\lambda$, the filter $h^{(m)}$ that leads to the minimal error bound. Then by a slight variation of (\ref{optim_ineq}), we obtain
 \begin{equation}\label{final-error}
\|e_\lambda\|_\infty \leq \|\varphi_0\|_1 \min\limits_{m\geq M(r)}  \left(\frac{e \sigma}{\pi}\right)^m m^m  \left(\frac{r_0}{r}\right)^m\lambda^{-m} \lesssim_r \exp\left({-\frac{\pi}{e^2 \sigma} \frac{r}{r_0} \lambda}\right)= 2^{-r\lambda},
\end{equation}
which proves the theorem.
 \end{proof}
 {\em Remark:} The smallest integer $\sigma$ such that the stability 
constraint $\|h\|_1\leq \gamma$ is satisfied for some $\gamma<2$ is $\sigma=6
$. In this case, Theorem~\ref{decayrate} yields exponential error decay 
for any rate constant $r<r_0\approx 0.102$. This is the fastest error decay 
currently known to be achievable for one-bit 
$\Sigma\Delta$ modulation. The previously 
best known bound for the achievable rate constant was $r_0\approx 0.088$~
\cite{Kra07}.

\section{Multi-level quantization alphabets and the case of
small $\sigma$}
\label{multi-level-section}

In this paper we have primarily considered one-bit $\Sigma\Delta$
modulators, though our theory and analysis is equally applicable
to quantization alphabets that consist of more than two levels.
Let us consider a general alphabet $\sA = \sA_L$ with $L$ levels 
such that
consecutive levels are separated by $2$ units as before. For instance,
$\sA_4 = \{-3, -1, 1, 3\}$ and $\sA_5 = \{-4, -2, 0, 2, 4\}$.
The corresponding generalized 
stability condition for the greedy quantization rule 
is then
\begin{equation}
 \|h\|_1 + \|y\|_\infty \leq L,
\end{equation}
which still guarantees the bound $\|v\|_\infty \leq 1$
(see, e.g., \cite[p. 104]{SchTe04}). 
In terms of our filter design and optimization problem, the parameter $\gamma$ that bounds $\|h\|_1$ can be set as large as $L$. 
The potential benefit of increased number of levels is that 
large values of $\gamma$ correspond to small values of $\sigma$,
which in turn yields faster exponential error decay rates
in the oversampling ratio $\lambda$
as given in \eqref{final-error}. There is, however, also 
an increase in the number of bits spent per sample which needs to
be accounted for if a rate-distortion type performance 
analysis is to be carried out.

As seen above
(Theorem \ref{thmasop} and Figure \ref{comparison}),
our analytical results are available for integer
values of $\sigma$ only. We will  
evaluate the performance of multi-level
quantization alphabets based on these values first. 
For each positive integer $\sigma$, we can find the minimum integer value of
$L \geq 2$ such that $\cosh (\pi/\sqrt{\sigma}) < L$. Alternatively,
if $L \geq 2$ is specified first, we 
find the minimum integer value of $\sigma$ that satisfies
this condition. Given such a $(\sigma,L)$ pair, the corresponding
bound $\|y \|_\infty$ on the input signal is
$L - \cosh(\pi/\sqrt{\sigma})$. We also compute the
achievable exponential error decay rate $r_0$ given by 
$\frac{\pi}{e^2 \sigma \ln 2}$, 
the average number of quantizer bits per sample
$B := \log_2 L$, and a coding efficiency figure
given by $r_0/B$. The results are tabulated in 
Table \ref{multi-level-comparison}. It turns out that 
the smallest value of $L$
which results in $\sigma = 1$ is $L=12$. This case also 
yields the highest coding efficiency figure given by $0.171$. It
also yields a significantly more favorable range for the input signal
compared to the one-bit case ($L = 2$).
As before, it is easy to check via Kolmogorov entropy bounds that 
the coding efficiency is always bounded by $1$. The cases $L=4$ and
$L=5$ are also noteworthy for they provide better overall performance,
yet still with a small quantization alphabet.

\begin{table}
\begin{center}
 \begin{tabular}{c|c|c|c|c|c}
$L$ & 2 & 3 & 4 & 5 & 12 \\ \hline
(average) number of bits/sample $B = \log_2 L$
& 1 & 1.585 & 2 & 2.322 & 3.585 \\ \hline
minimum integer value of $\sigma$ & 6 & 4 & 3 & 2 & 1 \\ \hline
maximum input signal $\|y\|_\infty$ & 0.058 &  0.490 & 0.851 & 0.335 & 0.408 \\ \hline
achievable error decay rate $r_0 = \pi(e^2 \sigma \ln 2)^{-1}$ & 0.102 & 0.153 & 0.204 & 0.306 & 0.613 \\ \hline
coding efficiency $r_0/B$ & 0.102 & 0.097 & 0.102 & 0.132 & 0.171
 \end{tabular}
\end{center}
\caption{Comparison of exponential error decay rates and
their coding efficiency for multi-level quantization alphabets. }
\label{multi-level-comparison}
\end{table}

It is natural to ask what can be said for 
the case of noninteger values of $\sigma$, especially as $\sigma
\to 0$. For our minimal subordinate constructions, we employ
the numerically computed values of 
$\lim_{m\to \infty} \frac{1}{m^2} 
\big(\eta({\bf n}^{(m)})\big)^{1/m}$ given in Figure \ref{comparison}
as a function of the continuous parameter $\sigma$.
As in Theorem \ref{decayrate},
the achievable error decay rate $r_0$ is given by 
\begin{equation}
 r_0 = \frac{1}{\pi e^2 \ln 2} 
\lim_{m\to \infty} \frac{m^2}{\big(\eta({\bf n}^{(m)})\big)^{1/m}}.
\end{equation}
On the other hand the smallest
number of levels $L$ of the quantizer 
which guarantees stability for the greedy rule is given by 
\begin{equation}
 L := \lceil \cosh (\pi / \sqrt{\sigma}) \rceil. 
\end{equation}
Finally, the coding efficiency as a function of $\sigma$ is
\begin{equation}\label{coding-efficiency}
 \frac{r_0}{B} = 
\frac{1}{\pi e^2 \ln \lceil \cosh (\pi / \sqrt{\sigma}) \rceil
} 
\lim_{m\to \infty} \frac{m^2}{\big(\eta({\bf n}^{(m)})\big)^{1/m}}.
\end{equation}
We plot our numerical findings in Figure \ref{efficiency}. The
specific values reported in Table \ref{multi-level-comparison}
are visible at the integer values $\sigma = 1,2,3,4,6$. It is 
interesting that the coding efficiency figure $0.171$ reported for
$\sigma = 1$ seems to be the global maximum value achievable by our
construction over all values of $\sigma$. This is possibly the
best performance achievable by minimally supported filters
that are constrained by subordinacy. We do not know if this figure
can be exceeded without the subordinacy constraint. 
The case of arbitrary optimal
filters (i.e., not necessarily minimally supported) 
is also open, along with the case of 
quantization rules that are more general than the
greedy rule. 

\begin{figure}[t] 
\centerline{\includegraphics[width=5in]{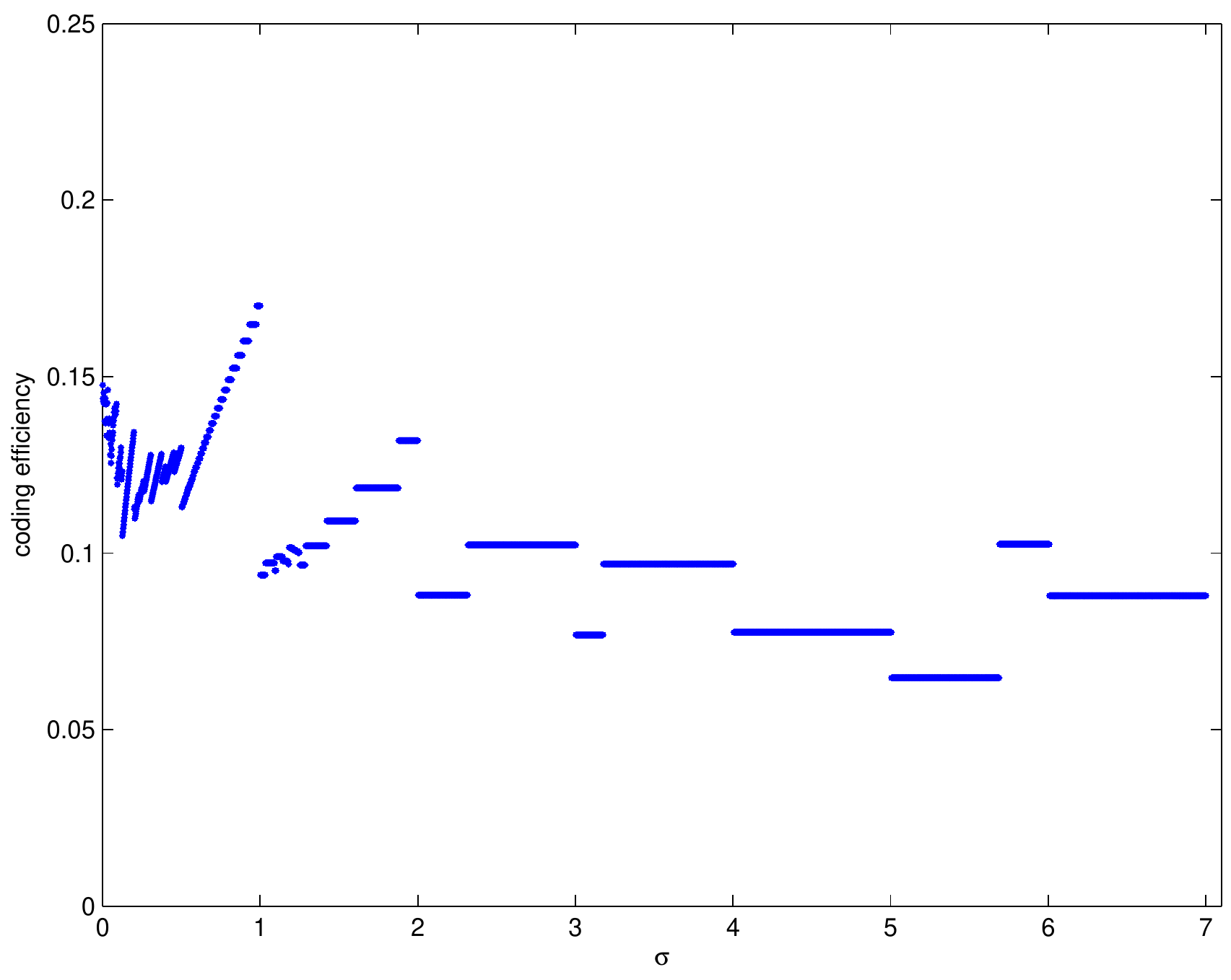}}
\caption{Numerical computation of the coding efficiency formula
given by \eqref{coding-efficiency}
for the minimal subordinate filters ${\bf n}^{(m)}$.}
\label{efficiency}
\end{figure}




\section*{Acknowledgements}
The work in this paper was supported in part by various grants and fellowships:
National Science Foundation Grants DMJ-0500923 (Deift), CCF-0515187 (G\"unt\"urk),
Alfred P. Sloan Research Fellowship (G\"unt\"urk),
the Morawetz Fellowship at the Courant Institute (Krahmer), the Charles M. Newman Fellowship at the Courant Institute (Krahmer)
and a NYU GSAS Dean's Student Travel Grant (Krahmer).

\appendix
\section{Some useful properties of Chebyshev polynomials}
Recall that the Chebyshev Polynomials of the first and second kind in $x=\cos\theta$ are given by
\begin {equation}
T_m(x)= \cos{m\theta}\ \ \ \ \ \ \ \text{and}\ \ \ \ \ \ \ U_m(x)= \frac{\sin(m+1)\theta}{\sin{\theta}},
\label{chebydef}
\end {equation}
respectively.
The Chebyshev polynomials have, in particular, the following properties (see \cite{szego39}, \cite{borerd95}):
 \begin{itemize}
\item $T_m'(x) = m U_{m-1}(x)$,
\item The zeros of $U_{m-1}$ are $z_j = \cos\left(\frac{m-j}{m}\pi \right)$, $j=1,\dots, m-1$,
\item For $m>0$, the leading coefficient of $T_m$ is $2^{m-1}$,
\item The Chebyshev polynomials satisfy the following identities
\begin{equation}
 T_m(\cosh\tau) = \cosh(m\tau), \ \ \ U_m(\cosh\tau)=\frac{\sinh(m\tau)}{\sinh\tau},
 \label{coshcheby}
\end{equation}
\item The Chebyshev polynomials satisfy the differential equation
\begin{equation}
(1-x)^2 T_m''(x)-x T_m'(x) +m^2 T_m(x)=0 \label{cpdiffeq}.
\end{equation}
\end{itemize}

We say that a  polynomial $p$ of degree $m$ has the {\em equi-oscillation property} on $[-1,1]$ (compare \cite{borerd95})  if it has $m-1$ real critical points $\zeta_1,\dots, \zeta_{m-1}$ which satisfy
\begin{equation}
 \zeta_0:=-1<\zeta_1<\dots< \zeta_{m-1}<\zeta_m:=1
\end{equation}
  such that the associated values are alternating 
\begin{equation}
 p(\zeta_j)=(-1)^{m-j}
\end{equation}
for $j=0,\dots, m$.

Note that if a polynomial has the equi-oscillation property then its leading coefficient is positive. The Chebyshev polynomials of the first kind $T_m$ have the equi-oscillation property for all $m$.  Indeed, the first two properties given above imply that the $z_j$'s are the critical points of $T_m$, and a simple calculation shows that $T_m(z_j)=(-1)^{m-j}$. The equi-oscillation property in fact characterizes the Chebyshev polynomials of the first kind:
\begin{proposition}\label{equosc}
  If $p(s)$ is a polynomial of degree $m$ in $s$ with the equi-oscillation property on $[-1,1]$, then $p=T_m$.
\end{proposition}

\begin{proof}
The proof follows ideas used in \cite{borerd95} to establish that, up to a constant, the $T_m$ are the unique monic polynomials with minimal $L^\infty$ norm.

Let $p(s)=a_p s^m+\dots$ and $q(s)=a_q s^m+\dots$ be two polynomials with the equi-oscillation property. W.l.o.g. assume $a_q\geq a_p>0$. Let $\zeta_1 <\dots<\zeta_{m-1}$ be the critical points of $p$ in $[-1,1]$ and set $\zeta_0=-1$, $\zeta_m=1$.  

Consider the  polynomial $r(s)=p(s)-\frac{a_p}{a_q}q(s)$ of degree $(m-1)$. Then $r(\zeta_{m-j})\geq 0$ for all even $j$, and $r(\zeta_{m-j})\leq 0$ for all odd $j$.  The proof that $r\equiv 0$ follows from the following more general statement:

{\textsc{Claim:}} {\em If $t_0<t_1<\dots <t_m\in \mathbb R$ and a polynomial $\rho$ of degree $m-1$ satisfies $(-1)^j\rho(\zeta_{j})\geq 0$ for all $j$, then $\rho\equiv 0$.}

We conclude that $r=p-\frac{a_p}{a_q} q\equiv 0$ by applying the claim to $\rho=(-1)^m r$. Since $p(1)=q(1)=1$ implies that $a_p=a_q$, we see that $p\equiv q$.
\end{proof}

{\em Proof of \textsc{Claim:}}
The proof proceeds by induction in $m$. In the case $m=1$, $\rho(t_0)\geq 0$ and $\rho(t_1)\leq 0$ implies that $r\equiv 0$. For the induction step, assume that the claim holds true for $m$. Given a polynomial $\rho$ of degree $m$ with the property, it must have a zero $z$ with $t_{m}\leq z \leq t_{m+1}$. Define $\tilde{\rho}(x)=\frac{r(x)}{z-x}$. Note that $\tilde{\rho}$ is a polynomial of degree $m-1$. If $z>t_m$, then $\tilde\rho(t_j)(-1)^j \geq 0$ for $0\leq j\leq m$ and hence $\tilde\rho\equiv 0$ by the induction hypothesis. If $z=t_m$ then $\tilde\rho(t_j)(-1)^j \geq 0$ for $0\leq j\leq m-1$, but clearly one also has $\tilde\rho(t_{m+1})(-1)^m \geq 0$. Again by the induction hypothesis $\tilde\rho\equiv 0$.
\hfill$\Box$

The following lemma plays a useful role in solving the relaxed minimization problem.
\begin{lemma}\label{sinprod}
 Let $z_j$, $j=1,\dots, m-1$, be the critical points of the Chebyshev polynomial of the first kind $T_m$, as above, and set $z_0\equiv -1$. 
Then
\begin{equation}
\prod\limits_{\substack{i=0\\i\neq k}}^{m-1}(z_k-z_i)=
\begin{cases}
  \frac{m(-1)^{m-1}}{2^{m-1}}\ \ \ \ \ \ \ \ \text{for $k=0$}\label{sinpr0}\\
  \\
  \frac{m(-1)^{m-1-k}}{2^{m-1}(1-z_k)}\ \ \ \ \ \text{for $k>0$}
\end{cases}
\end{equation}
\end{lemma}
\begin{proof}Recall that  $T_m$  has leading coefficient $2^{m-1}$. We obtain 
\begin{equation}
 T_m'(z) = m 2^{m-1} \prod\limits_{i=1}^{m-1}(z-z_i), \label{Tmprz0}
\end{equation}
and 
\begin{equation}
 T_m''(z) = m 2^{m-1} \sum_{j=1}^{m-1}\prod\limits_{\substack{i=1\\i\neq j}}^{m-1}(z-z_i),
\end{equation}
and hence for $1\leq k\leq m-1$
\begin{eqnarray}
 T_m''(z_k) &=& m 2^{m-1} \prod\limits_{\substack{i=1\\i\neq k}}^{m-1}(z_k-z_i)\\
&=& \frac{m 2^{m-1}}{1+z_k} \prod\limits_{\substack{i=0\\i\neq k}}^{m-1}(z_k-z_i).\label{Tmprprzk}
\end{eqnarray}
Thus for $1\leq k\leq m-1$ one has $T_m'(z_k)=0$ and $T_m(z_k)=(-1)^{m-k}$, and so (\ref{cpdiffeq}) reads
\begin{equation}
 (1-z_k^2)  \frac{m 2^{m-1}}{1+z_k} \prod\limits_{\substack{i=0\\i\neq k}}^{m-1}(z_k-z_i)+ m^2 (-1)^{m-k} = 0,
\end{equation}
or
\begin{equation}
 \prod\limits_{\substack{i=0\\i\neq k}}^{m-1}(z_k-z_i)=\frac{(-1)^{m-k-1}m}{2^{m-1}(1-z_k)}.
\end{equation}

On the other hand, as $z_0=\cos(\pi)$, we have using (\ref{chebydef})
\begin{equation}
\prod\limits_{i=1}^{m-1}(z_0-z_i)= \frac{1}{m2^{m-1}} T_m'(z_0) = \frac{1}{ 2^{m-1}} U_m(z_0) = \frac{1}{ 2^{m-1}}\lim\limits_{\theta\rightarrow\pi}\frac{ \sin (m\theta)}{\sin \theta} =\frac{(-1)^{m-1}m}{2^{m-1}}.
\end{equation}
\end{proof}

\frenchspacing

\bibliographystyle{plain}
\bibliography{dgk}

\begin{thebibliography}{10}

\bibitem{BPY}
J.~J. Benedetto, A.~M. Powell, and {\"O}.~Y{\i}lmaz.
\newblock Sigma-{D}elta ({$\Sigma\Delta$}) quantization and finite frames.
\newblock {\em IEEE Trans. Inform. Theory}, 52(5):1990--2005, 2006.

\bibitem{BoPa07}
B.~Bodmann and V.~I. Paulsen.
\newblock Frame paths and error bounds for sigma-delta quantization.
\newblock {\em Appl. Comput. Harmon. Anal.}, 22, 2007.

\bibitem{borerd95}
P.~Borwein and T.~Erd{\'e}lyi.
\newblock {\em Polynomials and polynomial inequalities}, volume 161 of {\em
  Graduate Texts in Mathematics}.
\newblock Springer-Verlag, New York, 1995.

\bibitem{CalDaub02}
A.~R. Calderbank and I.~Daubechies.
\newblock The pros and cons of democracy.
\newblock {\em IEEE Trans. Inform. Theory}, 48(6):1721--1725, 2002.
\newblock Special issue on Shannon theory: perspective, trends, and
  applications.

\bibitem{DaubDV03}
I.~Daubechies and R.~DeVore.
\newblock Reconstructing a bandlimited function from very coarsely quantized
  data: A family of stable sigma-delta modulators of arbitrary order.
\newblock {\em Ann. of Math}, 158:679--710, 2003.

\bibitem{DDGV06}
I.~Daubechies, R.~DeVore, C.~S. G{\"u}nt{\"u}rk, and V.~A. Vaishampayan.
\newblock A/{D} conversion with imperfect quantizers.
\newblock {\em IEEE Trans. Inform. Theory}, 52(3):874--885, 2006.

\bibitem{Gunt03}
C.~S. G{\"u}nt{\"u}rk.
\newblock One-bit sigma-delta quantization with exponential accuracy.
\newblock {\em Comm. Pure Appl. Math.}, 56:1608--1630, 2003.

\bibitem{Gunt04}
C.~S. G{\"u}nt{\"u}rk.
\newblock Approximating a bandlimited function using very coarsely quantized
  data: improved error estimates in sigma-delta modulation.
\newblock {\em J. Amer. Math. Soc.}, 17(1):229--242 (electronic), 2004.

\bibitem{GunturkNguyen}
C.~S. G{\"u}nt{\"u}rk and N.~T. Thao.
\newblock Ergodic dynamics in sigma-delta quantization: tiling invariant sets
  and spectral analysis of error.
\newblock {\em Adv. in Appl. Math.}, 34(3):523--560, 2005.

\bibitem{IYM62}
H.~Inose, Y.~Yasuda, and J.~Murakami.
\newblock A telemetering system by code manipulation - {$\Delta\Sigma$}
  modulation.
\newblock {\em IRE Trans on Space Electronics and Telemetry}, pages 204--209,
  1962.

\bibitem{Kra07}
F.~{Krahmer}.
\newblock {An improved family of exponentially accurate sigma-delta
  quantization schemes}.
\newblock In {\em Wavelets XII. Edited by Van De Ville, Dimitri; Goyal, Vivek
  K.; Papadakis, Manos. Proceedings of the SPIE}, volume 6701, October 2007.

\bibitem{NST96}
S.~R. Norsworthy, R.~Schreier, and G.~C. Temes, editors.
\newblock {\em Delta-Sigma-Converters: Theory, Design and Simulation}.
\newblock Wiley-IEEE, 1996.

\bibitem{SchTe04}
R.~Schreier and G.~C. Temes.
\newblock {\em Understanding Delta-Sigma Data Converters}.
\newblock Wiley-IEEE Press, 2004.

\bibitem{szego39}
G.~Szeg{\"o}.
\newblock {\em Orthogonal Polynomials}.
\newblock Amer. Math. Soc., 1939.

\bibitem{yilmaz02}
{\"O}.~Y{\i}lmaz.
\newblock Stability analysis for several sigma-delta methods of coarse
  quantization of bandlimited functions.
\newblock {\em Constructive Approximation}, 18(4):599--623, 2002.

\end{thebibliography}

\end{document}